\documentclass[12pt,onecolumn]{IEEEtran}

\IEEEoverridecommandlockouts

\usepackage{amsthm, amssymb}
\def\as{\xrightarrow{a.s.}} 
\def\cp{\xrightarrow{p.}} 
\def\cd{\xrightarrow{d.}} 

\def\b{\mathbb}

\def\sinr{\textrm{SINR}}
\def\snr{\textrm{SNR}}

\newtheoremstyle{slplain}
  {3pt}
  {3pt}
  {\slshape}
  {}
  {\bfseries}
  {.}%
  { }
  {}

\theoremstyle{slplain}

\newtheorem{cor}{Corollary}
\newtheorem{lem}{Lemma}
\newtheorem{pro}{Proposition}

\def\b{\mathbb}

\def\dint{{\rm d}}

\def\bg{{\mathbf{g}}}
\def\bh{{\mathbf{h}}}

\def\bq{{\mathbf{q}}}

\def\bu{{\mathbf{u}}}
\def\bv{{\mathbf{v}}}
\def\bw{{\mathbf{w}}}

\def\by{{\mathbf{y}}}

\def\b0{{\mathbf{0}}}

\def\bI{{\mathbf{I}}}

\def\bQ{{\mathbf{Q}}}

\def\bV{{\mathbf{V}}}

\def\bY{{\mathbf{Y}}}

\def\bepsilon{{\boldsymbol{\epsilon}}}

\usepackage{cite}
\usepackage{amsfonts}
\usepackage{multicol,multienum}
\usepackage{subfigure}

\ifCLASSINFOpdf
  \usepackage[pdftex]{graphicx}
  \graphicspath{{../eps/}{../jpeg/}}
  \DeclareGraphicsExtensions{.eps}
\else
  \usepackage{float}
  \usepackage[dvips]{graphicx}
  \graphicspath{{../eps/}}
  \DeclareGraphicsExtensions{.eps}
\fi

\usepackage[cmex10]{amsmath}

\usepackage{algorithm}
\usepackage{algorithmic}
\usepackage{array}
\usepackage{setspace}
\usepackage{url}
\usepackage{color}

\begin{document}

\title{   
The Interplay between Massive MIMO and Underlaid D2D Networking
}

\author{
\IEEEauthorblockA{Xingqin Lin, Robert W. Heath Jr., and Jeffrey G. Andrews}
\thanks{Xingqin Lin, Robert W. Heath Jr., and Jeffrey G. Andrews are with Department of Electrical $\&$ Computer Engineering, The University of Texas at Austin, USA. (Email: \{xlin, rheath\}@utexas.edu, jandrews@ece.utexas.edu).  
}
}

\maketitle

\begin{abstract}
In a device-to-device (D2D) underlaid cellular network, the uplink spectrum is reused by the D2D transmissions, causing mutual interference with the ongoing cellular transmissions. Massive MIMO is appealing in such a context as the base station's (BS's) large antenna array can nearly null the D2D-to-BS interference. The multi-user transmission in massive MIMO, however, may lead to increased cellular-to-D2D interference. This paper
studies the interplay between massive MIMO and underlaid D2D networking in a multi-cell setting. We investigate cellular and D2D spectral efficiencies under both perfect and imperfect channel state information (CSI) at the receivers that employ partial zero-forcing. 
Compared to the case without D2D, there is a loss in cellular spectral efficiency due to D2D underlay. With perfect CSI, the loss can be completely overcome if the number of canceled D2D interfering signals is scaled with the number of BS antennas at an arbitrarily slow rate. With imperfect CSI, in addition to pilot contamination, a new asymptotic effect termed \textit{underlay contamination} arises.
In the non-asymptotic regime, simple analytical lower bounds are derived for both the cellular and D2D spectral efficiencies.
\end{abstract}

\section{Introduction}

\subsection{Background}

Device-to-device (D2D) communication enables nearby mobile devices to establish direct links in cellular networks \cite{3gppD2dRF, Corson2012toward, lin2013overview}, unlike traditional cellular communication where all traffic is routed via base stations (BSs). D2D has the potential to improve spectrum utilization, shorten packet delay, and reduce energy consumption, while enabling new peer-to-peer and location-based applications and services \cite{Fodor2012Design, lin2013overview} and being a required feature in public safety networks \cite{doumi2013lte}. Introducing D2D poses many challenges and risks to the existing cellular architecture. In particular, in a D2D underlaid cellular network where the spectrum is reused D2D transmission may cause interference to cellular transmission and vice versa. Existing operator services may be severely affected if the newly introduced D2D interference is not appropriately controlled.


The distinctive traits of massive MIMO make it appealing to enable D2D communication in the uplink resources of cellular networks. In a massive  MIMO system, each BS uses a very large antenna array to serve multiple users in each time-frequency resource block \cite{marzetta2010noncooperative}. If the number of antennas at a BS is significantly larger than the number of served users, the channel of each user to/from the BS is nearly orthogonal to that of any other user. This allows for very simple transmit or receive processing techniques like matched filtering to be nearly optimal with enough antennas even in the presence of interference \cite{marzetta2010noncooperative, ngo2011energy, hoydis2013massive, bai2013asymptotic, madhusudhanan2013stochastic}. This implies that, with a large antenna array at a BS, D2D signals possibly result in close-to-zero interference at the uplink massive MIMO BS, making D2D very simple and appealing in massive MIMO systems. 

Though D2D-to-cellular interference may be effectively handled by the large antenna array at a BS, cellular-to-D2D interference persists and may be worse in a massive MIMO system. Specifically, massive MIMO is a multi-user transmission strategy designed to support multiple users in each time-frequency block; the number of simultaneously active uplink users is scalable with the number of antennas at the BS. With this increased number of uplink transmitters, the D2D links reusing uplink radio resources will experience increased interference. To protect D2D links, the number of simultaneously active uplink users might have to be limited, eating into massive MIMO gain. It is not \textit{a priori} clear to what extent the D2D signals would be affected by the multiuser transmission and the tradeoff between supporting D2D communication and scaling up the uplink capacity in a massive MIMO system. Further, if cochannel D2D signals are present when estimating massive MIMO channels, the estimated channel state information (CSI) would become less accurate, which may hurt massive MIMO performance. It is not \textit{a priori} clear however to what extent the D2D signals would affect the channel estimation and consequently the performance of the massive MIMO system.

Existing research on D2D networking is mainly focused on single-antenna networks (see, e.g., \cite{xu2010effective, yu2011resource, kaufman2013spectrum, lin2013multicast, ji2013fundamental, lin2013comprehensive}) while research on the use of antenna arrays has just begun \cite{janis2009interferenceMIMO, tang2013cooperative, li2012device, jayasinghe2013mimo, min2011capacity}. To mitigate or avoid mutual interference between cellular and D2D transmissions, \cite{janis2009interferenceMIMO, tang2013cooperative} considered precoding while \cite{li2012device, jayasinghe2013mimo} studied various relaying strategies. In contrast, \cite{min2011capacity} proposed not to schedule uplink users that may generate excessive interference to D2D users. How D2D MIMO and cellular MIMO interact, especially in the massive MIMO context, is still largely open.

\subsection{Contributions and Outcomes}

The main contributions and outcomes of this paper are summarized as follows.

\subsubsection{A tractable hybrid network model}

We introduce a tractable hybrid network model consisting of both ad hoc nodes and cellular infrastructure, which extends our previous single-antenna D2D model \cite{lin2013comprehensive, lin2013multicast} to multi-antenna transmission. We consider a multi-cell setting and focus on the uplink which is better than the downlink for D2D underlay \cite{lin2013overview}.
The spatial positions of the underlaid D2D transmitters are modeled by a Poisson point process (PPP). Such a random PPP model is well motivated by the random and unpredictable D2D user locations \cite{andrews2011tractable, baccelli2012optimizing}.  All the transmissions (both cellular and D2D) in this model are SIMO (i.e., single-input multiple-output) with each BS having a very large antenna array. For the receive processing, we extend the partial zero-forcing (PZF) receiver studied in ad hoc networks \cite{jindal2011multi} to the hybrid network in question. Spectral efficiency is used as the sole metric throughout this paper.

\subsubsection{Spectral efficiency with perfect CSI}

In the asymptotic regime where the number of BS antennas $M \to \infty$ and with perfect CSI, we find that the received signal-to-interference-plus-noise ratio (SINR) of any cellular user increases unboundedly and the effects of noise, fast fading, and the interfering signals from the other co-channel cellular users and the \textit{infinite} D2D transmitters vanish completely. Equivalently, it is possible to reduce cellular transmit power as $\Theta(1/M)$ but still achieve a non-vanishing cellular spectral efficiency, as in the case without D2D underlay \cite{ngo2011energy}. Compared to the case without D2D, with scaled cellular transmit power $\Theta(1/M)$, there is a loss in cellular spectral efficiency if a constant number of D2D interfering signals is canceled. The loss can be overcome if the number of canceled D2D interfering signals is scaled appropriately (e.g., $\Theta(\sqrt{M})$). In the non-asymptotic regime, we derive simple analytical lower bounds for both cellular and D2D spectral efficiencies; the derived bounds allow for very efficient numerical evaluation.

\subsubsection{Spectral efficiency with imperfect CSI}

We study pilot-based CSI estimation in which known training sequences are transmitted and the receivers use minimum mean squared error (MMSE) estimators for channel estimation. In the asymptotic regime with the estimated CSI, it is known that the received SINR of any cellular user is bounded due to pilot contamination \cite{marzetta2010noncooperative}. With D2D underlay, the bounded SINR is further degraded due to a 
new asymptotic effect which we term \textit{underlay contamination}. Due to the underlay contamination, we find that scaling down cellular transmit power results in a vanishing cellular spectral efficiency, no matter how slow the scaling rate is. This is dramatically different from the case without D2D underlay, for which \cite{ngo2011energy} shows that cellular transmit power can be scaled down as $\Theta(1/\sqrt{M})$. To recover the power scaling law $\Theta(1/\sqrt{M})$, one possible approach is to deactivate the D2D links in the training phase of massive MIMO; however, compared to the case without D2D,  there is a loss in cellular spectral efficiency due to  D2D-to-cellular interference in the data transmission phase. Instead, if the cellular transmit power is not scaled down and D2D links are deactivated in the training phase, massive MIMO automatically eliminates the effect of D2D-to-cellular interference in the data transmission phase.


%
%
%
%

\section{Mathematical Models}
\label{sec:model}

\subsection{Network Model}

Consider a multi-cell D2D underlaid massive MIMO system shown in Fig. \ref{fig:1}. In this system, there are $B+1$ cells; in each cell $b, b = 0,1,...,B$, $K$ cellular user equipments (UEs) transmit to the BS $b$. We denote by $\mathcal{K}_{b}$ the set of the $K$ cellular UEs in the cell $b$, and $\mathcal{C}_b$ the coverage area of the cell $b$ satisfying that $\mathcal{C}_b \cap \mathcal{C}_{b'} = \emptyset, \forall b \neq b'$. We assume that the $K$ cellular UEs are uniformly distributed in each cell; this assumption is not essential in the analysis but will be used in the simulation.
Specifically, as spatial division multiple access would be challenging for cellular UEs of high mobility, it makes more sense to consider static or low-mobility scenarios. Therefore, we condition on cellular UE positions when studying the achievable performance of a particular cellular link. But we still average over all possible realizations of cellular UE positions in the simulation to compute an average overall performance.

The cellular system is underlaid with D2D UEs. The locations of D2D transmitters are distributed as a homogeneous PPP $\Phi$ with density $\lambda$, as they are random and unpredictable. We partition $\Phi$ into $B+2$ disjoint PPPs $\Phi_0,...,\Phi_{B+1}$, where
$
\Phi_b = \Phi \cup \mathcal{C}_b, \forall b=0,...,B, \textrm{ and } \Phi_{B+1} = \Phi \backslash \cup_{i=0}^B \Phi_i .
$
Each D2D receiver is located at a random distance of $D$ meters from its  associated D2D transmitter with uniformly distributed direction. 

We focus on SIMO in this paper, i.e., a transmitter (either cellular or D2D) uses one antenna for transmission, while a BS and a D2D receiver respectively use $M$ and $N$ antennas for receiving. The analysis and results in this paper can be extended to MIMO transmission, i.e., spatial multiplexing, by treating a UE with multiple data streams as multiple co-located virtual UEs, each sending one data stream. We are interested in the performance regime where $M$ is large and the assumption $M \gg K$ is made throughout this paper, as in the seminal work on massive MIMO \cite{marzetta2010noncooperative}. Note that the scenario where the ratio $K/M$ converges to some constant has also been widely assumed when studying the asymptotic behaviors of MIMO performance \cite{tulino2004random, couillet2011random}. Studying this scenario is an interesting topic, which we leave to future work.

In this system, all the transmitters use the same time-frequency resource block, leading to cochannel interference.
We assume that cellular and D2D UEs transmit at constant powers $P_{\textrm{c}}$ and $P_{\textrm{d}}$ respectively.

\begin{figure}
\centering
\includegraphics[width=14cm]{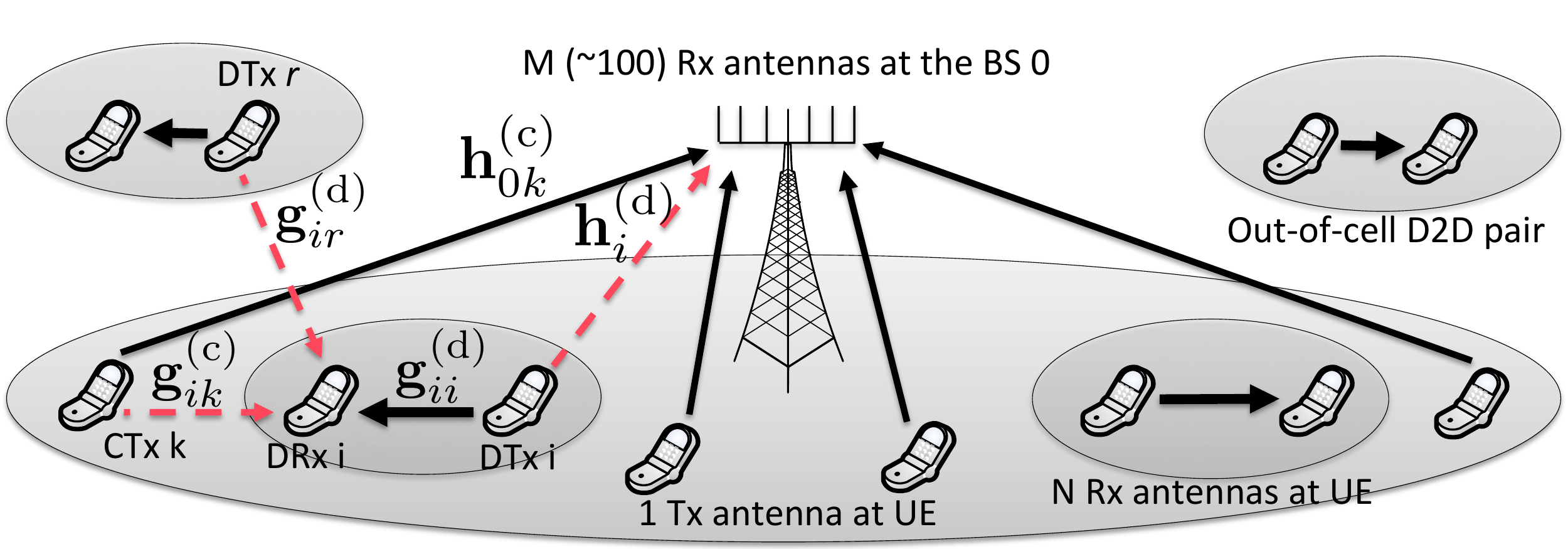}
\caption{A D2D underlaid massive MIMO system consisting of both cellular and D2D links. For clarity, we only show the central cell. D2D pairs located outside of the cells are out of cellular coverage but still contribute to the total aggregate D2D interference.}
\label{fig:1}
\end{figure}

\subsection{Baseband Channel Models}


Without loss of generality, we focus on the central cell, whose BS is indexed by $b = 0$ and located at the origin. This helps simplify the notation.
The  $M \times 1$ dimensional baseband received signal at the central BS is
\begin{align}
\by^{(\textrm{c})}_0 = \sum_{b=0}^{B} \sum_{k\in \mathcal{K}_b} \sqrt{P_\textrm{c} \Xi^{(\textrm{c})}_{bk}} \|x^{(\textrm{c})}_{bk}\|^{-\frac{\alpha_{\textrm{c}}}{2}}  \bh^{(\textrm{c})}_{bk} u^{(\textrm{c})}_{bk}  + \sum_{i \in \Phi} \sqrt{P_\textrm{d}\Xi_i^{(\textrm{d})}} \|x^{(\textrm{d})}_i\|^{-\frac{\alpha_{\textrm{c}}}{2}}  \bh_i^{(\textrm{d})} u_i^{(\textrm{d})}  + \bv^{(\textrm{c})}_0 ,
\label{eq:1}
\end{align}
where $\Xi^{(\textrm{c})}_{bk}$ denotes the shadowing from cellular transmitter $k$ in the cell $b$ to the BS $0$, $x_{bk}^{(\textrm{c})}$ denotes the position of cellular transmitter $k$ in the cell $b$, $\alpha_\textrm{c} >2$ denotes the pathloss exponent of UE-BS links, $\bh_{bk}^{(\textrm{c})} \in \mathcal{C}^{M \times 1}$ is the vector channel  from cellular transmitter $k$ in the cell $b$ to the BS $0$, $u_{bk}^{(\textrm{c})}$ denotes the zero-mean unit-variance transmit symbol of cellular transmitter $k$ in the cell $b$, $\Xi_i^{(\textrm{d})}, x_i^{(\textrm{d})},  \bh_i^{(\textrm{d})} \in \mathcal{C}^{M \times 1}$ and $u_i^{(\textrm{d})}$ are similarly defined for D2D transmitter $i$,
and $\bv_0^{(\textrm{c})} \in \mathcal{C}^{M \times 1}$ is  complex Gaussian noise at the BS $0$ with covariance $N_0 \bI_M$, where  $\bI_M$ denotes the $M$ dimensional identity matrix. 


Similarly, the  $N \times 1$ dimensional baseband received signal at D2D receiver $r$ is
\begin{align}
\by_r^{(\textrm{d})} = \sum_{b=0}^{B} \sum_{k\in \mathcal{K}_b } \sqrt{P_\textrm{c}\Xi^{(\textrm{c})}_{rbk}} (d_{rbk}^{(\textrm{c})})^{-\frac{\alpha_\textrm{d}}{2}}  \bg_{rbk}^{(\textrm{c})} u_{bk}^{(\textrm{c})}  + \sum_{i \in \Phi} \sqrt{P_\textrm{d}\Xi_{ri}^{(\textrm{d})}} (d_{ri}^{(\textrm{d})})^{-\frac{\alpha_\textrm{d}}{2}}  \bg_{ri}^{(\textrm{d})}  u_i^{(\textrm{d})}  + \bv_r^{(\textrm{d})} ,
\label{eq:2}
\end{align}
where  $\Xi^{(\textrm{c})}_{rbk}, \Xi_{ri}^{(\textrm{d})}$ are the shadowing from cellular transmitter $k$ in the cell $b$ to D2D receiver $r$ and from D2D transmitter $i$ to D2D receiver $r$ respectively, $d_{rbk}^{(\textrm{c})} \triangleq \|x^{(\textrm{c})}_{bk} - z_r^{(\textrm{d})}\|$ and $d_{ri}^{(\textrm{d})} \triangleq \|x^{(\textrm{d})}_i - z_r^{(\textrm{d})}\|$ with $z_r^{(\textrm{d})}$ denoting the position of D2D receiver $r$,  $\alpha_\textrm{d}>2$ denotes the pathloss exponent of UE-UE links, $\bg_{rbk}^{(\textrm{c})}, \bg_{ri}^{(\textrm{d})} \in \mathcal{C}^{N \times 1}$ are the vector channels from  cellular transmitter $k$ in the cell $b$ to D2D receiver $r$ and from  D2D transmitter $i$ to   D2D receiver $r$ respectively,  and $\bv_r^{(\textrm{d})} \in \mathcal{C}^{N \times 1}$ is  complex Gaussian noise with covariance $N_0 \bI_N$. 

Note that we have used different pathloss exponents $\alpha_\textrm{c}$ and $\alpha_\textrm{d}$ for UE-BS and UE-UE links (cf. (\ref{eq:1}) and (\ref{eq:2})) due to their different propagation characteristics. Specifically, the antenna height of a macro BS is tens of meters, while the typical antenna height at a UE is under 2 m. As a result, both terminals of a UE-UE link are low and see similar near street scattering environment, which is different from the radio environment around a macro BS \cite{lin2013overview}.

In this paper, we assume Gaussian signaling, i.e., $\{u_{bk}^{(\textrm{c})}\}, \{u_{i}^{(\textrm{d})}\}$ are i.i.d. complex Gaussian $\mathcal{CN}(0,1)$,  and i.i.d. shadowing with mean $\bar{\Xi}$.
We also assume that all the vector channels have i.i.d. $\mathcal{CN}(0,1)$ elements, independent across transmitters. It follows that the \textit{favorable propagation} condition \cite{rusek2013scaling} desired in massive MIMO systems holds in our model:
\[ \frac{1}{M} \bh_{br}^{(s)*} \bh_{b'\ell}^{(s')} \as \left\{ \begin{array}{ll}
         1 & \mbox{if $s=s'$, $b=b'$ and $r = \ell$};\\
         0 & \mbox{otherwise},\end{array} \right. \] 
where $s \in \{\textrm{c},\textrm{d}\}$, $\as$ denotes the almost sure convergence as $M \to \infty$, and when $s = \textrm{d}$ the first subindex $b$ in $\bh_{br}^{(s)}$ should be understood as null. Recent measurement campaigns have given evidence to validate favorable propagation for massive MIMO in practice \cite{larsson2013massive}.

\subsection{Receive Filters}

Denote by $\bw^{(\textrm{c})}_k$ the filter used by the central BS for receiving the signal of cellular transmitter $k$ in the central cell, i.e., the central BS detects the symbol $u_{0k}^{(\textrm{c})}$ based on $\bw^{(\textrm{c})*}_k \by_0^{(\textrm{c})}$. Similarly,  D2D receiver $r$ detects the symbol $u_r^{(\textrm{d})}$ based on $\bw^{(\textrm{d})*}_r \by_r^{(\textrm{d})}$, where $\bw^{(\textrm{d})}_r$ denotes the filter used by D2D receiver $r$. The performance of the D2D underlaid massive MIMO system depends on the receive filters. In general, either  the receive filters can be designed to boost desired signal power or they can be used to cancel undesired interference. In this paper, we focus on a particular type of linear filters: the PZF receiver, which uses a subset of the degrees of freedom for boosting received signal power and the remainder for interference cancellation.

The central BS uses $m_\textrm{c}$ and $m_\textrm{d}$ degrees of freedom to cancel the interference from the \textit{nearest} $m_\textrm{c}$ cellular interferers and the nearest $m_\textrm{d}$ D2D interferers. A feasible choice of $(m_\textrm{c},m_\textrm{d})$ needs to be in the following set:
\begin{align}
\mathcal{Z}_{\textrm{c}} = \{ (m_\textrm{c}, m_\textrm{d}) \in \mathbb{N} \times \mathbb{N}: m_\textrm{c} \leq (B+1)K-1, m_\textrm{c} + m_\textrm{d} \leq M-1 \}.
\end{align}
The PZF filter $\bw^{(\textrm{c})}_k$ is the projection of the channel vector $\bh_{0k}^{(\textrm{c})}$ onto the subspace orthogonal to the one spanned by the channel vectors of canceled interferers.
For ease of reference, we denote by $\mathcal{K}^{(\textrm{c})}_{bk}$ the set of \textit{uncanceled} cellular interferers in the cell $b$ and $\Phi^{(\textrm{c})}_k$ the set of  uncanceled D2D interferers when detecting the symbol $u^{(\textrm{c})}_{0k}$ of cellular transmitter $k$ in the central cell.

Similarly, each D2D receiver uses $n_\textrm{c}$ and $n_\textrm{d}$ degrees of freedom to cancel the interference from the nearest $n_\textrm{c}$ cellular interferers and the nearest $n_\textrm{d}$ D2D interferers, and $(n_\textrm{c},n_\textrm{d})$ needs to be in the following set:
\begin{align}
\mathcal{Z}_{\textrm{d}} = \{ (n_\textrm{c}, n_\textrm{d}) \in \mathbb{N} \times \mathbb{N}: n_\textrm{c} \leq (B+1)K, n_\textrm{c} + n_\textrm{d} \leq N-1 \}.
\end{align}
The PZF filter $\bw^{(\textrm{d})}_r$ of D2D receiver $r$ is the projection of the channel vector $\bg_{rr}^{(\textrm{d})}$ onto the subspace orthogonal to the one spanned by the channel vectors of canceled interferers.
For ease of reference, we denote by $\mathcal{K}^{(\textrm{d})}_{br}$ the set of uncanceled cellular interferers in the cell $b$ and $\Phi^{(\textrm{d})}_r$ the set of uncanceled D2D interferers at D2D receiver $r$.

\textbf{Remark on PZF receiver.}
Although suboptimal, PZF receivers have several advantages that motivate us to focus on them in this paper. On the one hand, PZF receivers are relatively general: they reduce to maximum ratio combining (MRC) receivers when $m_\textrm{c}+m_\textrm{d}=0$ and $n_\textrm{c}+n_\textrm{d}=0$ and to conventional fully ZF receivers when $m_\textrm{c}+m_\textrm{d}=M-1$ and $n_\textrm{c}+n_\textrm{d}=N-1$. PZF receivers are conceptually similar to MMSE in that they use the degrees-of-freedom for both interference suppression and signal enhancement. It has been shown that PZF receivers can achieve the same scaling law in terms of transmission capacity as MMSE receivers \cite{jindal2011multi}, which is not true for either MRC or fully ZF receivers. On the other hand, PZF receivers are analytically more tractable than other more sophisticated receivers like MMSE receivers from a system point of view. This analytical tractability allows us to develop an explicit characterization of the performance of the massive MIMO system with D2D underlay. 
Nevertheless, as noted in \cite{jindal2011multi}, MMSE fitlers should be used in practice because they have less stringent CSI requirements while being the optimal linear filters. Specifically, a PZF receiver requires CSI of each canceled interferer, while a MMSE receiver only requires knowledge of the covariance of the aggregate interference. The stringent CSI required by a BS's PZF may be satisfied; we will discuss how UE-BS channels can be estimated in Section \ref{sec:imperfect}. In contrast, it may be too expensive to satisfy the stringent CSI requirement of a UE's PZF. To cancel the nearest interferers, a D2D receiver has to know the positions of the interferers; this requirement may be relaxed to some extent by using a PZF that cancels the \textit{strongest} interferers instead of the nearest interferers. Further, we need to design reference signals for estimating UE-UE channels and the network should coordinate UE-UE channel estimation process. Considering these complexities, a general PZF at a D2D receiver is not practical. Instead, MRC may be preferred, or a MMSE receiver can be used.

\section{Spectral Efficiency with Perfect Channel State Information}
\label{sec:perfect}

In this section, we derive the spectral efficiency of cellular and D2D links under the assumption of perfect CSI; the case of imperfect CSI will be treated in the next section.

\subsection{Asymptotic Cellular Spectral Efficiency}

For cellular UE $k$ in the central cell, the  post-processing SINR with the PZF filter $\bw_k^{(\textrm{c})}$ is 
\begin{align}
\sinr_k^{(\textrm{c})} = \frac{ S^{(\textrm{c})}_k }{ I^{(\textrm{c}\to \textrm{c})}_k +  I^{(\textrm{d}\to \textrm{c})}_k  + \| \bw_k^{(\textrm{c})}  \|^2 N_0  }, 
\end{align}
where $S^{(\textrm{c})}_k = P_\textrm{c} \Xi_{0k}^{(\textrm{c})} \|x_{0k}^{(\textrm{c})}\|^{-\alpha_\textrm{c}}  \| \bw_{k}^{(\textrm{c})*} \bh_{0k}^{(\textrm{c})} \|^2$ denotes the desired signal power of cellular UE $k$, $I^{(\textrm{c}\to \textrm{c})}_k$ and $ I^{(\textrm{d}\to \textrm{c})}_k$ respectively denote the cochannel cellular and D2D interference powers experienced by cellular UE $k$ and are given by
\begin{align}
I^{(\textrm{c}\to \textrm{c})}_k &=\sum_{b=0}^B \sum_{\ell \in \mathcal{K}^{(\textrm{c})}_{bk}  } P_\textrm{c} \Xi_{b\ell}^{(\textrm{c})} \|x_{b\ell}^{(\textrm{c})} \|^{-\alpha_\textrm{c}} |\bw_k^{(\textrm{c})*} \bh_{b\ell}^{(\textrm{c})} |^2   \\  
I^{(\textrm{d}\to \textrm{c})}_k &= \sum_{i \in \Phi_k^{(\textrm{c})}}  P_\textrm{d} \Xi_i^{(\textrm{d})} \|x_i^{(\textrm{d})}\|^{-\alpha_\textrm{c}} |\bw_k^{(\textrm{c})*}  \bh_i^{(\textrm{d})} |^2 .
\label{eq:19}
\end{align}
The spectral efficiency of cellular UE $k$ in the central cell is defined as
\begin{align}
R^{(\textrm{c})}_{k} = \mathbb{E} \left[ \log ( 1 + \sinr_k^{(\textrm{c})}  ) \right] , 
\end{align}
where the expectation is taken with respect to the fast fading, shadowing and random locations of UEs.

\begin{pro}
With perfect CSI and conditioned on $\{x_{bk}^{(\textrm{c})}\}$, as $M \to \infty$, the desired signal power $S^{(\textrm{c})}_k$ when normalized by $M^2$ and conditioned on $\Xi_{0k}^{(\textrm{c})}$ converges to
\begin{align}
&\lim_{M \to \infty} \frac{1}{M^2}  S^{(\textrm{c})}_k \as  P_\textrm{c} \Xi_{0k}^{(\textrm{c})} \|x_{0k}^{(\textrm{c})}\|^{-\alpha_\textrm{c}}, 
\end{align}
and the cellular interference power $I^{(\textrm{c}\to \textrm{c})}_k$, the D2D interference power $I^{(\textrm{d}\to \textrm{c})}_k$, and the noise power $\| \bw_k^{(\textrm{c})}  \|^2 N_0$ when normalized by $M^2$ converge as follows.  
\begin{align}
&\lim_{M \to \infty} \frac{1}{M^2}  I^{(\textrm{c}\to \textrm{c})}_k \as  0, \quad \lim_{M \to \infty} \frac{1}{M^2}  I^{(\textrm{d}\to \textrm{c})}_k \cp  0,  \quad \lim_{M \to \infty} \frac{1}{M^2} \| \bw_k^{(\textrm{c})}  \|^2 N_0 \as  0 ,
\end{align}
where $\cp$ denotes the convergence in probability. 
\label{pro:1}
\end{pro}
\begin{proof}
See Appendix \ref{proof:pro:1}.
\end{proof}

Prop. \ref{pro:1} shows that with perfect CSI, as $M\to \infty$, the post-processing $\sinr_k^{(\textrm{c})}$ increases unboundedly in probability (as almost sure convergence implies convergence in probability). More specifically, a deterministic received power of the desired signal from cellular UE $k$ (conditioned on its pathloss and shadowing) can be achieved and the effects of noise, fast fading, and the interfering signals from the other $K-1$ cellular UEs and the \textit{infinite} D2D transmitters vanish completely. Therefore, Prop. \ref{pro:1} validates the intuition that with perfect CSI D2D-to-cellular interference can be made arbitrarily small with a large enough antenna array at the BS. Note that the D2D-to-cellular interference can be completely nulled out, even though (i) the number of the PPP distributed D2D interferers is infinite and (ii) the mean of the aggregate D2D interference is infinite. Further, the proof of Prop. \ref{pro:1} shows that a simple MRC filter with $m_\textrm{c}=m_\textrm{d}=0$ suffices. 

Though Prop. \ref{pro:1} shows that arbitrarily large received SINR can be achieved with massive MIMO, this in practice will ultimately fail since the received power cannot be larger than the transmit power. Further, it may not be possible to fully exploit a very high SINR due to  practical constraints such as the highest order of modulation and coding schemes. Nevertheless, the large array gains may be translated into power savings for cellular UEs: with a given SNR target we can lower the transmit powers of cellular UEs and thus improve their energy efficiency, as shown in the following proposition.

\begin{pro}
With perfect CSI, fixed PZF parameters $(m_\textrm{c}, m_\textrm{d})$, scaled cellular transmit power ${P_\textrm{c}}/{M}$, and conditioned on $\Xi_{0k}^{(\textrm{c})}$ and $x_{0k}^{(\textrm{c})}$, as $M \to \infty$, the spectral efficiency $R^{(\textrm{c})}_{k}$ of cellular UE $k$ in the central cell converges to
\begin{align}
R^{(\textrm{c})}_{k}  &\to \mathbb{E}_{\Phi, \eta} \left[ \log \left( 1 + \frac{ \snr_{0k}^{(\textrm{c})} }{ \sum_{i \in \Phi_k^{(\textrm{c})}}   \frac{P_\textrm{d}}{N_0}  \|x_i^{(\textrm{d})}\|^{-\alpha_\textrm{c}} \eta_i  + 1 }  \right)  \right], \label{eq:14} 
\end{align}
where $\snr_{bk}^{(\textrm{c})} =  {P_\textrm{c} \Xi_{bk}^{(\textrm{c})}  \|x_{bk}^{(\textrm{c})}\|^{-\alpha_\textrm{c}}  }/{N_0}$, $\{\eta_i\}$ are i.i.d. random variables distributed as $\eta_i \sim \textrm{Exp}(1)$. Further, if $m_\textrm{d}+1 > \frac{\alpha_\textrm{c}}{2}$,
\begin{align}
\lim_{M\to \infty} R^{(\textrm{c})}_{k}  
&\geq \log \left( 1 + \frac{ \snr_{0k}^{(\textrm{c})} }{ \rho ( m_\textrm{d}, \alpha_\textrm{c} )  + 1 }  \right),
\label{eq:15}
\end{align}
where
\begin{align}
\rho ( m, \alpha ) = \frac{2 (\pi \lambda)^{\frac{\alpha}{2}} P_\textrm{d}\bar{\Xi} \Gamma ( m + 1 -\frac{\alpha}{2} ) }{ (\alpha - 2) N_0 \Gamma(m)}   ,
\label{eq:13}
\end{align}
where the Gamma function $\Gamma (x) = \int_0^{\infty} t^{x-1} e^{-t} \dint t$.
\label{pro:4}
\end{pro}
\begin{proof}
See Appendix \ref{proof:pro:4}.
\end{proof}

Note that in Prop. \ref{pro:4}, if the underlaid D2D transmitters did not exist, the spectral efficiency $R^{(\textrm{c})}_{k}$ of cellular UE $k$  (conditioned on its pathloss and shadowing)  in the central cell would converge to $\log \left( 1 +    \snr_{0k}^{(\textrm{c})} \right)$, the maximum achievable spectral efficiency of a point-to-point SISO (single-input single-output) Gaussian channel. It is as if massive MIMO could simultaneously support $K$ interference-free SISO links while reducing the power of each cellular UE by $10\log_{10}M$ dB. This result is consistent with Prop. 1 in \cite{ngo2011energy} without D2D underlay.

With D2D underlay, the asymptotic result (\ref{eq:14}) shows that there is a loss in cellular spectral efficiency due to the uncanceled interfering signals from the D2D transmitters in $\Phi_k^{(\textrm{c})}$, i.e., D2D transmitters in $\Phi$ except the nearest $m_\textrm{d}$ ones whose signals are canceled by the PZF filter. Though it is possible to derive an exact analytical expression (involving integrals) for (\ref{eq:14}), we give a more intuitive lower bound  (\ref{eq:15}), which succinctly characterizes the loss due to the D2D underlay through a single term $\rho ( m_\textrm{d}, \alpha_\textrm{c} )$. Several remarks are in order.

\textbf{Remark 1.} The term $\rho ( m_\textrm{d}, \alpha_\textrm{c} )$ corresponding to the uncanceled D2D interference increases with $P_\textrm{d}$ and $\lambda$ and decreases with $m_\textrm{d}$, agreeing with intuition: larger transmit power or larger density of D2D interferers or smaller number of canceled D2D interferers leads to higher D2D-to-cellular interference, thus lowering the cellular spectral efficiency. Further, $\rho ( m_\textrm{d}, \alpha_\textrm{c} ) \sim \lambda^{\frac{\alpha_c}{2}}$, because a linear increase in $\lambda$ implies that the distances of the PPP distributed D2D transmitters to the BS decrease as $\lambda^{\frac{1}{2}}$ and thus the D2D-to-cellular interference power increases as $\lambda^{\frac{\alpha_\textrm{c}}{2}}$.

\textbf{Remark 2.} Note that the lower bound (\ref{eq:15}) is meaningful only if $m_\textrm{d} + 1 > \frac{\alpha_\textrm{c}}{2}$. 
As $m_\textrm{d} \to \frac{\alpha_\textrm{c}}{2} - 1$, $\Gamma(m_\textrm{d}+1 - \frac{\alpha_\textrm{c}}{2}) \to \infty$ and thus $\rho ( m_\textrm{d}, \alpha_\textrm{c} ) \to \infty$. In fact, from the proof of Prop. \ref{pro:4}, we can see that $\rho ( m_\textrm{d}, \alpha_\textrm{c} )$ denotes the mean residual D2D-to-cellular interference after canceling the $m_\textrm{d}$ nearest D2D intererers. If $m_\textrm{d} \leq \frac{\alpha_\textrm{c}}{2} - 1$, the mean residual D2D-to-cellular interference is infinite, and the lower bound (\ref{eq:15}) becomes $0$, which is trivially true.

Next we show that the loss of cellular spectral efficiency due to D2D underlay can be recovered if we scale the number $m_\textrm{d}$ of canceled D2D interferers to infinity as $M \to \infty$. Further, the growth rate of $m_\textrm{d}$ can be arbitrarily slow. 

\begin{pro}
With perfect CSI, arbitrary but fixed $m_\textrm{c}$, scaled cellular transmit power ${P_\textrm{c}}/{M}$, and conditioned on $\Xi_{0k}^{(\textrm{c})}$ and $x_{0k}^{(\textrm{c})}$, if $m_\textrm{d} \in \omega (1)$ and $m_\textrm{d} \in o(M)$, i.e., $m_\textrm{d}$ increases to infinity at a rate slower than $\Theta(M)$, the spectral efficiency $R^{(\textrm{c})}_{k}$ of cellular UE $k$ in the central cell converges as follows.
\begin{align}
R^{(\textrm{c})}_{k}  &\to  \log \left( 1 +   \snr_{0k}^{(\textrm{c})}  \right), \quad \textrm{ as } M \to \infty.
\end{align}
\label{pro:5}
\end{pro}
\vspace{-2em}
\begin{proof}
According to Stirling's formula, $\Gamma (t+1) \sim \sqrt{2\pi t} (\frac{t}{e})^t$, when $t$ is large. It follows that
\begin{align}
\frac{\Gamma ( m_\textrm{d}  + 1 -\frac{\alpha_\textrm{c} }{2} ) }{\Gamma ( m_\textrm{d}   )} &\sim   \frac{\sqrt{2\pi ( m_\textrm{d}   -\frac{\alpha_\textrm{c} }{2} )} (\frac{m_\textrm{d}   -\frac{\alpha_\textrm{c} }{2}}{e})^{m_\textrm{d}   -\frac{\alpha_\textrm{c} }{2}}}{\sqrt{2\pi ( m_\textrm{d}   - 1 )} (\frac{m_\textrm{d}   - 1}{e})^{m_\textrm{d}   -1}}
\notag \\
&
= \left( \frac{e}{ m_\textrm{d}   -\frac{\alpha_\textrm{c} }{2} } \right)^{\frac{\alpha_\textrm{c} }{2}-1} \left( \frac{m_\textrm{d}   -\frac{\alpha_\textrm{c} }{2}}{m_\textrm{d}   -1} \right)^{m_\textrm{d}   -\frac{ 1 }{2}} \sim \left( \frac{1}{ m_\textrm{d}   -\frac{\alpha_\textrm{c} }{2} } \right)^{\frac{\alpha_\textrm{c} }{2}-1} .
\end{align}
Therefore, as $m_\textrm{d} \to \infty$, $\rho ( m_\textrm{d}, \alpha_\textrm{c} ) \to 0$ and thus
\begin{align}
\log \left( 1 +   \snr_{0k}^{(\textrm{c})}    \right) \geq \lim_{M\to \infty} R^{(\textrm{c})}_{k}  
&\geq \log \left( 1 + \frac{ \snr_{0k}^{(\textrm{c})} }{ \rho ( m_\textrm{d}, \alpha_\textrm{c} )  + 1 }  \right) \to \log \left( 1 +   \snr_{0k}^{(\textrm{c})}    \right) .
\end{align}
This completes the proof.
\end{proof}


Before ending this subsection, we would like to point out that the power scaling law $\Theta(1/M)$ does not imply that cellular UEs can transmit at a vanishing power level as $M$ increases to infinity. A more appropriate understanding is that large power savings are possible when $M$ is large, and there exists a limiting cellular spectral efficiency. For example, numerical results (e.g., Fig. \ref{fig:4}) show that when $M$ is on the order of several hundred, the limiting cellular spectral efficiency is reached and the power savings are about 20 - 30 dB. Increasing $M$ further theoretically leads to more power savings, but the conclusion becomes fragile since non-ideal effects like spurious emission may not be negligible when ${P_\textrm{c}}/{M}$ becomes too small. The power scaling law for imperfect CSI case studied in Section \ref{sec:imperfect} should also be understood in a similar fashion.

\subsection{Non-asymptotic Cellular Spectral Efficiency}

Next we analyze the cellular spectral efficiency in the non-asymptotic regime to generate more insights into the impact of the various system parameters. To this end, using Jensen's inequality we derive a lower bound for $R^{(c)}_{k}$ in the following proposition.

\begin{pro}
With perfect CSI, $M \geq m_\textrm{c} + m_\textrm{d} +1$ and $m_\textrm{d} > \frac{\alpha_\textrm{c}}{2} - 1$, and conditioned on $\{\Xi_{bk}^{(\textrm{c})}\}$ and $\{x_{bk}^{(\textrm{c})}\}$, the spectral efficiency $R^{(\textrm{c})}_{k}$ of cellular UE $k$ in the central cell is lower bounded as
\begin{align}
R^{(\textrm{c})}_{k} \geq  R^{(\textrm{c,lb})}_{k} = \log \left( 1 + \frac{  (M - m_\textrm{c} - m_\textrm{d}-1) \snr_{0k}^{(\textrm{c})} }{\sum_{b=0}^B \sum_{\ell \in \mathcal{K}^{(\textrm{c})}_{bk}  } \snr_{b\ell}^{(\textrm{c})}  + \rho ( m_\textrm{d}, \alpha_\textrm{c} )   + 1 }  \right) , 
\label{eq:8}
\end{align}
where $\rho(m,\alpha)$ is defined in (\ref{eq:13}).
\label{pro:2}
\end{pro}
\begin{proof}
See Appendix \ref{proof:pro:2}.
\end{proof}

Note that the first term in the denominator of (\ref{eq:8}) corresponds to the uncanceled cellular interference; it decreases as $m_\textrm{c}$ increases. Similarly, the second term in the denominator of (\ref{eq:8}) corresponds to the uncanceled D2D interference; it  decreases as $m_\textrm{d}$ increases. In contrast, the numerator of (\ref{eq:8}) corresponds to the desired signal power; it decreases as $m_\textrm{c}$ and/or $m_\textrm{d}$ increase. The lower bound (\ref{eq:8}) demonstrates the various tradeoffs when choosing the PZF parameters $m_\textrm{c}$ and $m_\textrm{d}$. Note that such tradeoffs disappear in the asymptotic regime (cf. Prop. \ref{pro:4} and \ref{pro:5}). 

We point out that the received signal power gain is only proportional to $M-m_\textrm{c}-m_\textrm{d}-1$ in the lower bound (\ref{eq:8}). One might think the power gain should be proportional to $M-m_\textrm{c}-m_\textrm{d}$, the number of degrees of freedom left for power boosting after using $m_\textrm{c} + m_\textrm{d}$  degrees of freedom  for interference cancellation. The fallacy of the above argument is that it ignores the effect of fading, which makes a power gain proportional to $M-m_\textrm{c}-m_\textrm{d}$ unachievable.

We may optimize the PZF filter $\bw^{(\textrm{c})}_k$ by choosing  ($m^\star_\textrm{c}, m^\star_\textrm{d}$) such that they maximize the sum spectral efficiency in the central cell, i.e., 
\begin{align}
(m^\star_\textrm{c}, m^\star_\textrm{d}) = \textrm{argmax}_{(m_\textrm{c}, m_\textrm{d}) \in \mathcal{Z}_{\textrm{c}}}  \sum_{k=1}^K R^{(\textrm{c})}_{k} .
\end{align}
This is a combinatorial optimization, and finding the global optimum involves exhaustive search over the feasible space $\mathcal{Z}_{\textrm{c}}$. In practice, each BS only cancels intra-cell cellular interference, leading to  $0\leq m_\textrm{c} \leq K-1$. Further, existing studies (see, e.g., \cite{jindal2011multi}) show that canceling a few Poisson distributed interferers provides close-to-optimal performance. This implies that it suffices to consider a few small values for  $m_\textrm{d}$. These two facts greatly reduce the search space for ($m^\star_\textrm{c}, m^\star_\textrm{d}$). A demonstrative numerical result is given in Fig. \ref{fig:8} in Section \ref{sec:sim}.

\subsection{D2D Spectral Efficiency}

For D2D receiver $r$, the post-processing SINR  with the PZF filter $\bw_r^{(\textrm{d})}$ is 
\begin{align}
\sinr_r^{(\textrm{d})} = \frac{ S^{(\textrm{d})}_r }{ I^{(\textrm{c}\to \textrm{d})}_r +  I^{(\textrm{d}\to \textrm{d})}_r  + \| \bw_r^{(\textrm{d})}  \|^2 N_0  }, 
\end{align}
where $S^{(\textrm{d})}_r = P_\textrm{d} \Xi_{rr}^{(\textrm{d})} (d_{rr}^{(\textrm{d})})^{-\alpha_\textrm{d}}  \| \bw_r^{(\textrm{d})*} \bg_{rr}^{(\textrm{d})} \|^2$ denotes the desired signal power of D2D Tx-Rx pair $r$, $I^{(\textrm{c}\to \textrm{d})}_r$ and $ I^{(\textrm{d}\to \textrm{d})}_r$ respectively denote the cochannel cellular and D2D interference powers experienced by D2D receiver $r$ and are given by
\begin{align}
I^{(\textrm{c}\to \textrm{d})}_r &= \sum_{b=0}^B \sum_{k \in \mathcal{K}^{(\textrm{d})}_{br}  } P_\textrm{c} \Xi_{rbk}^{(\textrm{c})} (d_{rbk}^{(\textrm{c})})^{-\alpha_\textrm{d}} |\bw_r^{(\textrm{d})*} \bg_{rbk}^{(\textrm{c})} |^2 \notag \\
I^{(\textrm{d}\to \textrm{d})}_r &= \sum_{i \in \Phi_r^{(\textrm{d})}}  P_\textrm{d} \Xi_{ri}^{(\textrm{d})} (d_{ri}^{(\textrm{d})})^{-\alpha_\textrm{d}} |\bw_r^{(\textrm{d})*}  \bg_{ri}^{(\textrm{d})} |^2 .
\end{align}
The spectral efficiency of D2D Tx-Rx pair $r$ is defined as
\begin{align}
R^{(\textrm{d})}_{r} = \mathbb{E} \left[ \log ( 1 + \sinr_r^{(\textrm{d})}  ) \right] , 
\end{align}
where the expectation is taken with respect to the fast fading, shadowing and random locations of UEs.

As the number $N$ of antennas at UEs is often limited due to hardware constraints, it is not very meaningful to study the asymptotic performance with $N \to \infty$. Instead, as in the case of cellular spectral efficiency, we provide a lower bound for $R^{(\textrm{d})}_{r}$ in the non-asymptotic regime, which characterizes the impact of the various system parameters on the D2D spectral efficiency.

\begin{pro}
With perfect CSI, $N \geq n_\textrm{c} + n_\textrm{d} +1$ and $n_\textrm{d} > \frac{\alpha_\textrm{d}}{2} - 1$, and conditioned on $\Xi_{rr}^{(\textrm{d})}$, $d_{rr}^{(\textrm{d})}$, $\{\Xi_{bk}^{(\textrm{c})}\}$ and $\{x_{bk}^{(\textrm{c})}\}$, the spectral efficiency $R^{(\textrm{d})}_{r}$ of D2D Tx-Rx pair $r$ is lower bounded as
\begin{align}
R^{(\textrm{d})}_{r} \geq  R^{(\textrm{d,lb})}_{r} = \log \left( 1 + \frac{  (N - n_\textrm{c} - n_\textrm{d}-1) \snr_r^{(\textrm{d})} }{\sum_{b=0}^B\sum_{k \in \mathcal{K}^{(\textrm{d})}_{br}  } \frac{P_\textrm{c}}{N_0} \Xi_{rbk}^{(\textrm{c})}(d_{rbk}^{(\textrm{c})})^{-\alpha_\textrm{d}}  + \rho(n_\textrm{d},\alpha_\textrm{d})   + 1 }  \right) , 
\label{eq:12}
\end{align}
where $\snr_r^{(\textrm{d})} =  {P_\textrm{d} \Xi_{rr}^{(\textrm{d})}  (d_{rr}^{(\textrm{d})})^{-\alpha_\textrm{d}}  }/{N_0}$, and $\rho(m,\alpha)$ is defined in (\ref{eq:13}).
\label{pro:3}
\end{pro}
\begin{proof}
The proof is similar to that of Prop. \ref{pro:2} and is omitted for brevity.
\end{proof}

Many of the remarks on Prop. \ref{pro:2} apply to Prop. \ref{pro:3} as well and are not repeated here. One additional remark is that the cellular-to-D2D interference is not homogeneous: the D2D receivers located in the boundary of the cellular network experience less cellular interference than the D2D receivers located in the central cell. But if we focus on the D2D performance in the central cell and choose the number of cellular cells large enough, this heterogeneity can be made negligible.

\section{Spectral Efficiency with Imperfect Channel State Information}
\label{sec:imperfect}


\subsection{Estimating UE-BS Channels}

We consider pilot-based CSI estimation in which known training sequences are transmitted and used for estimation purpose. To alleviate the training overhead and coordination complexity, we assume that each BS $b$ does not estimate the channels from other-cell transmitters (either cellular or D2D). Note that as the number $|\Phi_b|$ of D2D transmitters in the cell $b$ is Poisson distributed, there may be less than $m_\textrm{d}$ D2D transmitters in the cell $b$. Therefore, during the training phase, each BS $b$ requires the $K$ cellular UEs and the $m_{\textrm{d},b} \triangleq \min(m_\textrm{d}, |\Phi_b|)$ nearest D2D transmitters (w.r.t. the BS $b$) in its cell to simultaneously transmit orthogonal training sequences. The BSs do not coordinate the other D2D transmitters, which can send independent symbols during the training phase.

Unlike the perfect CSI case, other-cell transmissions (both cellular and D2D) now have a more delicate impact on the performance of the central cell. To accommodate this, in this subsection we extend the previous notation as follows. We add an additional subscript $b$ to $x_i^{(\textrm{d})}, \Xi_i^{(\textrm{d})}$ and $\bh_i^{(\textrm{d})}$, and obtain $x_{bi}^{(\textrm{d})}, \Xi_{bi}^{(\textrm{d})}$ and $\bh_{bi}^{(\textrm{d})}$, indicating that they are associated with D2D transmitter $i$ in the cell $b$. Similarly, we use $\Phi_{bk}^{(\textrm{c})}$ to denote the set of uncanceled D2D interferers in the cell $b, b=0,...,B+1$. Note that the coverage of the ``cell'' $B+1$ is simply the complement (w.r.t. $\mathbb{R}^2$) of the coverage areas of the cells $0,...,B$, and the ``cell'' $B+1$ does not contain a BS.

Denoting by $T_\textrm{c} \geq K+ m_\textrm{d}$ the length of a training sequence, we can represent the training sequences as a $T_\textrm{c} \times (K+m_\textrm{d})$ dimensional matrix $\sqrt{T_\textrm{c}}\bQ^{(\textrm{c})} = \sqrt{T_\textrm{c}} (\bq^{(\textrm{c})}_1,...,\bq^{(\textrm{c})}_{K+m_\textrm{d}})$ satisfying $\bQ^{(\textrm{c})*} \bQ^{(\textrm{c})} = \bI_{K+m_\textrm{d}}$. These pilots are reused over different cells.
In the training phase, the  $M \times T_\textrm{c}$ dimensional baseband received signal $\bY^{(\textrm{c})}_0$ at the central BS  is
\begin{align}
\bY_0^{(\textrm{c})} =& \sum_{b=0}^B \sum_{k\in \mathcal{K}_b} \sqrt{T_\textrm{c} P_\textrm{c} \Xi^{(\textrm{c})}_{bk}}  \|x^{(\textrm{c})}_{bk}\|^{-\frac{\alpha_\textrm{c}}{2}} \bh^{(\textrm{c})}_{bk} \bq^{(\textrm{c})*}_k  + \sum_{b=0}^B \sum_{i=1}^{m_{\textrm{d},b}}\sqrt{T_\textrm{c} P_\textrm{c} \Xi^{(\textrm{d})}_{bi}}  \|x^{(\textrm{d})}_{bi}\|^{-\frac{\alpha_\textrm{c}}{2}} \bh_{bi}^{(\textrm{d})} \bq^{(\textrm{c})*}_{K+i} \notag \\
&+ \sum_{b=0}^{B+1} \sum_{r \in \Phi_{bk}^{(\textrm{c})}} \sqrt{P_\textrm{d}\Xi^{(\textrm{d})}_{br}}  \|x^{(\textrm{d})}_{br}\|^{-\frac{\alpha_\textrm{c}}{2}} \bh_{br}^{(\textrm{d})} \bu_{br}^{(\textrm{d})*}  + \bV_0^{(\textrm{c})} ,
\label{eq:16}
\end{align}
where  the  $T_\textrm{c} \times 1$ dimensional vector $\bu_{br}^{(\textrm{d})}$ contains the data symbols sent by D2D interferer $r$ in the cell $b$, and the $M \times T_\textrm{c}$ dimensional noise matrix $\bV_0^{(\textrm{c})}$ consists of i.i.d. $\mathcal{CN}(0,N_0)$ elements. Note that the coordinated D2D transmitters also use power $P_\textrm{c}$ during the training phase since they now transmit to their associated BSs. 

We assume that the central BS uses linear MMSE estimator for the channel estimation. To this end, we first project the received signal $\bY_0^{(\textrm{c})}$ in the direction of $\bq_{\tilde{k}}^{(\textrm{c})}$ and normalize it to obtain
\begin{align}
\tilde{\by}^{(s)}_k &= \frac{1}{\sqrt{T_\textrm{c} P_\textrm{c} \Xi^{(s)}_{0k}}  \|x^{(s)}_{0k}\|^{-\frac{\alpha_\textrm{c}}{2}}}\bY_0^{(\textrm{c})} \bq_{\tilde{k}}^{(\textrm{c})} \notag \\
&= \bh^{(s)}_{0k}  +  \sum_{b=1}^B \sqrt{\beta^{(s)}_{bk}} \bh^{(s)}_{bk}  +  \tilde{\bv}_k^{(s)}, \quad (s, \tilde{k} ) \in \{ (\textrm{c},k), (\textrm{d}, K + k) \},
\end{align}
where 
\[ \beta^{(s)}_{bk}  \triangleq \left\{ \begin{array}{ll}
         0 & \mbox{if $s =$ d and $k > m_{\textrm{d},b}$ };\\
                  \frac{\Xi^{(s)}_{bk}  \|x^{(s)}_{bk}\|^{-\alpha_\textrm{c}}}{\Xi^{(s)}_{0k}  \|x^{(s)}_{0k}\|^{-\alpha_\textrm{c}} } & \mbox{otherwise}, \end{array} \right. \] 
and        
$\tilde{\bv}_k^{(s)}$ denotes the equivalent channel estimation ``noise'' and is given by
\begin{align}
\tilde{\bv}_k^{(s)} = \frac{1}{\sqrt{T_\textrm{c} P_\textrm{c} \Xi^{(s)}_{0k}}  \|x^{(s)}_{0k}\|^{-\frac{\alpha_\textrm{c}}{2}}}   \left( \sum_{b=0}^{B+1} \sum_{r \in \Phi_{bk}^{(\textrm{c})}} \sqrt{P_\textrm{d}\Xi^{(\textrm{d})}_{br}}  \|x^{(\textrm{d})}_{br}\|^{-\frac{\alpha_\textrm{c}}{2}} \bh_{br}^{(\textrm{d})} \bar{u}_{br}^{(\textrm{d})}  + \bar{\bv}_k^{(\textrm{c})} \right) .
\label{eq:18}
\end{align}
where $\bar{u}_{br}^{(\textrm{d})} = \bu_{br}^{(\textrm{d})*} \bq_{\tilde{k}}^{(\textrm{c})}$ and $\bar{\bv}_k^{(\textrm{c})} = \bV_0^{(\textrm{c})} \bq_{\tilde{k}}^{(\textrm{c})}$.

\begin{lem}
The linear MMSE estimate of $\bh^{(s)}_{0k}, s\in \{\textrm{c}, \textrm{d}\}$, is given by $\hat{\bh}^{(s)}_{0k} = \xi_k^{(s)} \tilde{\by}^{(s)}_k$, where 
\begin{align}
\xi_k^{(s)}
&= \left(1 + \sum_{b=1}^B  \beta^{(s)}_{bk} + \frac{\sum_{b=0}^{B+1} \sum_{r \in \Phi_{bk}^{(\textrm{c})}} P_\textrm{d}\Xi^{(\textrm{d})}_{br}  \|x^{(\textrm{d})}_{br}\|^{-\alpha_\textrm{c}}    + N_0 }{ T_\textrm{c} P_\textrm{c} \Xi^{(s)}_{0k} \|x^{(s)}_{0k}\|^{-\alpha_\textrm{c}} } \right)^{-1}  .
\label{eq:30}
\end{align}
Further, $\mathbb{E}[\hat{\bh}^{(s)}_{0k}] = 0$ and $\mathbb{E}[\hat{\bh}^{(s)}_{0k} \hat{\bh}^{(s)*}_{0k}] =\xi_k^{(s)} \bI_M$. As for the estimation error $\bepsilon_k^{(s)} = \bh^{(s)}_{0k} - \hat{\bh}^{(s)}_{0k}$, $\mathbb{E}[\bepsilon_k^{(s)}] = 0$ and $\mathbb{E}[\bepsilon_k^{(s)} \bepsilon_k^{(s)*}] = (1 - \xi_k^{(s)})  \bI_M$.
\label{lem:1}
\end{lem}
\begin{proof}
See Appendix \ref{proof:lem:1}.
\end{proof}

Lemma \ref{lem:1} shows that the longer the length $T_\textrm{c}$ of a training sequence, the smaller the covariance of the estimation error $\bepsilon_k^{(s)}$ and thus the more accurate the channel estimation $\hat{\bh}^{(s)}_{0k}$, agreeing with intuition. In particular, $\mathbb{E}[\bepsilon_k^{(s)} \bepsilon_k^{(s)*}] \to \frac{\sum_{b=1}^B  \beta^{(s)}_{bk}}{1+\sum_{b=1}^B  \beta^{(s)}_{bk}} \bI_M$, as $T_\textrm{c} \to \infty$. This shows that even with infinitely long training sequences the channel estimation cannot be perfect due to pilot contamination.

\subsection{Asymptotic Cellular Spectral Efficiency}

In this subsection, we examine the asymptotic performance of the cellular links as $M \to \infty$. For simplicity, we focus on $m_\textrm{c} = m_\textrm{d} =0$. Then $\bw^{(\textrm{c})}_k= \hat{\bh}_{0k}^{(\textrm{c})}$ is the MRC filter. Since multiplying the filter by a constant does not affect the post-processing SINR, we may choose $\bw^{(\textrm{c})}_k= \bY_0^{(\textrm{c})} \bq_{k}^{(\textrm{c})}$.
It follows that $\lim_{M\to \infty} \frac{1}{M}\bw^{(\textrm{c})*}_k \by^{(\textrm{c})}_0$ equals
\begin{align}
&\lim_{M\to \infty} \frac{1}{M}\left(\sum_{b=0}^B \sqrt{T_\textrm{c} P_\textrm{c} \Xi^{(\textrm{c})}_{bk}}  \|x^{(\textrm{c})}_{bk}\|^{-\frac{\alpha_\textrm{c}}{2}} \bh^{(\textrm{c})}_{bk}  +     \left( \sum_{i \in \Phi} \sqrt{P_\textrm{d}\Xi^{(\textrm{d})}_{i}}  \|x^{(\textrm{d})}_{i}\|^{-\frac{\alpha_\textrm{c}}{2}} \bh_{i}^{(\textrm{d})} \bar{u}_i^{(\textrm{d})}  + \bar{\bv}^{(\textrm{c})}_0 \right) \right)^*\by^{(\textrm{c})}_0 \notag \\
&=   \sum_{b=0}^{B} \sqrt{T_\textrm{c}}  P_\textrm{c} \Xi^{(\textrm{c})}_{b k} \|x^{(\textrm{c})}_{b k}\|^{-\alpha_{\textrm{c}}} u^{(\textrm{c})}_{b k}  + \sum_{i \in \Phi} P_\textrm{d}\Xi_i^{(\textrm{d})} \|x^{(\textrm{d})}_i\|^{-\alpha_{\textrm{c}}}  \bar{u}_i^{(\textrm{d})*}  u_i^{(\textrm{d})}   .
\label{eq:50}
\end{align}

The first term in (\ref{eq:50}) is the usual phenomenon  appearing in massive MIMO \cite{marzetta2010noncooperative}. In particular, it indicates that asymptotically the effects of uncorrelated receiver noise and fast fading vanish, and there is no intra-cell interference. The remaining effect is the residual other-cell interference due to pilot reuse across the cells \cite{marzetta2010noncooperative}.  With D2D underlay, we observe that a new effect (i.e., the last term in (\ref{eq:50})) indicating the residual D2D-to-cellular interference arises. The reason why the effect of D2D underlay does not vanish can be explained as follows. The interfering signal of D2D transmitter $i$ in the training phase correlates with the interfering signal of D2D transmitter $i$ in the data transmission phase through the common channel vector $\bh_{i}^{(\textrm{d})}$. Therefore, unlike the uncorrelated receiver noises in the estimation phase and in the data transmission phase, when multiplying the estimated channel $\hat{\bh}_{0k}^{(\textrm{c})}$ with the received signal $\by^{(\textrm{c})}_0$, the effect of D2D underlay cannot be eliminated even with infinitely many antennas at the BS. We term this effect \textit{underlay contamination}.

Note that the D2D underlay contamination term in (\ref{eq:50}) involves the products of complex Gaussian random variables $\bar{u}_i^{(\textrm{d})*}  u_i^{(\textrm{d})} $, the D2D interfering signals are not Gaussian distributed. It is known that given a covariance constraint Gaussian noise is the worst-case noise for additive noise channels. Therefore, treating the D2D interfering signals as Gaussian noises, we obtain the following Lemma \ref{lem:2}.
\begin{lem}
With imperfect CSI at the central BS and $(m_\textrm{c}, m_\textrm{d}) =(0, 0)$, i.e., the MRC receiver $\bw^{(\textrm{c})}_k= \hat{\bh}_{0k}^{(\textrm{c})}$, the following spectral efficiency $\hat{R}^{(\textrm{c})}_{k}$ is achievable for cellular UE $k$ in the central cell.
\begin{align}
\hat{R}^{(\textrm{c})}_{k} = \mathbb{E} \left[ \log \left( 1 + \frac{ \hat{S}^{(\textrm{c})}_k }{\hat{I}^{(\textrm{c}\to \textrm{c})}_k + \hat{I}^{(\textrm{d}\to \textrm{c})}_k} \right)  \right] ,
\label{eq:52}
\end{align}
where $\hat{S}^{(\textrm{c})}_k =  T_\textrm{c}  P^2_\textrm{c} |\Xi^{(\textrm{c})}_{0 k}|^2 \|x^{(\textrm{c})}_{0 k}\|^{-2\alpha_{\textrm{c}}}$, and
\begin{align}
\hat{I}^{(\textrm{c}\to \textrm{c})}_k = \sum_{b=1}^B T_\textrm{c}  P^2_\textrm{c} |\Xi^{(\textrm{c})}_{b k}|^2 \|x^{(\textrm{c})}_{b k}\|^{-2\alpha_{\textrm{c}}} , \quad 
\hat{I}^{(\textrm{d}\to \textrm{c})}_k = \sum_{i \in \Phi} P^2_\textrm{d} |\Xi_i^{(\textrm{d})}|^2 \|x^{(\textrm{d})}_i\|^{-2\alpha_{\textrm{c}}}.
\end{align}
\label{lem:2}
\end{lem}
\vspace{-2em}

Unlike the perfect CSI case in which the SINR of a cellular link can be made arbitrarily large (c.f. Prop. \ref{pro:1}), Lemma \ref{lem:2} shows that with imperfect CSI there is a limit on the received SINR in massive MIMO due to the pilot contamination and D2D underlay contamination.
With D2D underlay, conditioned on UE positions and shadowing, the loss of SINR of cellular UE $k$ in the central cell is 
$
10\log_{10} (1+ \hat{I}^{(\textrm{d}\to \textrm{c})}_k/\hat{I}^{(\textrm{c}\to \textrm{c})}_k ) 
$
dB. There are four possible approaches to mitigate the loss. First, we can decrease the D2D transmit power. This approach reduces the link budgets of D2D links, limiting the range of D2D communication. Second, we can increase the cellular transmit power. This approach increases the energy consumption of cellular UEs and also results in more cellular-to-D2D interference. Third, we can increase the length of training sequences. But longer training sequences consume more cellular transmission resources in terms of both power and bandwidth. Fourth, we can deactivate the D2D links in the training phase of massive MIMO. Then we retain the usual asymptotic cellular spectral efficiency in massive MIMO:
\begin{align}
\hat{R}^{(\textrm{c})}_{k} = \mathbb{E} \left[ \log \left( 1 + \frac{  |\Xi^{(\textrm{c})}_{0 k}|^2 \|x^{(\textrm{c})}_{0 k}\|^{-2\alpha_{\textrm{c}}} }{\sum_{b=1}^B  |\Xi^{(\textrm{c})}_{b k}|^2 \|x^{(\textrm{c})}_{b k}\|^{-2\alpha_{\textrm{c}}} }  \right)  \right] .
\end{align}  
Certainly, the last approach reduces time resources for D2D communication.

The following Corollary \ref{cor:1} shows that with D2D underlay contamination it is impossible to scale down cellular transmit powers, and thus D2D underlay hurts the energy efficiency of cellular UEs in massive MIMO. 
\begin{cor}
Scaling down cellular transmit powers results in a vanishing cellular spectral efficiency, i.e., $\hat{R}^{(\textrm{c})}_{k} \to 0$, as $P_\textrm{c} \to 0$.
\label{cor:1}
\end{cor}

To achieve a non-vanishing cellular spectral efficiency while scaling down cellular transmit powers, one solution is to schedule two independent sets of active D2D transmitters in the estimation phase and in the data transmission phase of massive MIMO. This solves underlay contamination. The disadvantage is that the BSs cannot use the estimated D2D UE-BS channels in the estimation phase to cancel the interference from the other set of D2D transmitters in the data transmission phase. Therefore, its performance is not clear in the non-asymptotic regime. Another simple solution is to deactivate the D2D links in the training phase of massive MIMO. Then we can scale down cellular transmit powers as in the following Prop. \ref{pro:6}.
\begin{pro}
With D2D links deactivated in the training phase of massive MIMO and scaled cellular transmit power ${P_\textrm{c}}/\sqrt{M}$, as $M \to \infty$, the achievable spectral efficiency $\hat{R}^{(\textrm{c})}_{k}$ of cellular UE $k$ in the central cell converges as follows.
\begin{align}
\hat{R}^{(\textrm{c})}_{k}  &\to \mathbb{E} \left[ \log \left( 1 + \frac{T_\textrm{c} ( \snr_{0k}^{(\textrm{c})} )^2 }{  \sum_{b=1}^B T_\textrm{c} ( \snr_{bk}^{(\textrm{c})} )^2  + \sum_{i \in \Phi} \frac{P_\textrm{d}}{N_0} \Xi_i^{(\textrm{d})} \|x^{(\textrm{d})}_i\|^{-\alpha_{\textrm{c}}} +1}  \right)  \right]  .
\label{eq:20}
\end{align}
\label{pro:6}
\end{pro}
\begin{proof}
See Appendix \ref{proof:pro:6}.
\end{proof}

Finally, we give a more explicit expression for the asymptotic cellular spectral efficiency to allow for efficient numerical evaluation.
\begin{pro}
The achievable spectral efficiency $\hat{R}^{(\textrm{c})}_{k}$ of cellular UE $k$ in the central cell given in (\ref{eq:52}) equals
\begin{align}
\hat{R}^{(\textrm{c})}_{k} 
&= \int_0^\infty \frac{1}{z} (1- \mathbb{E} [e^{-z \hat{S}^{(\textrm{c})}_k}]) \mathbb{E} [e^{-z\hat{I}^{(\textrm{c}\to \textrm{c})}_k}] \mathbb{E} [ e^{-z \hat{I}^{(\textrm{d}\to \textrm{c})}_k)}] \dint z ,
\end{align}
where $\mathbb{E} [e^{-z \hat{S}^{(\textrm{c})}_k}] = \mathbb{E} [e^{-z T_\textrm{c}  P^2_\textrm{c} |\Xi^{(\textrm{c})}_{0 k}|^2 \|x^{(\textrm{c})}_{0 k}\|^{-2\alpha_{\textrm{c}}}}]$, $\mathbb{E} [e^{-z\hat{I}^{(\textrm{c}\to \textrm{c})}_k}] = \prod_{b=1}^B \mathbb{E} [e^{-z T_\textrm{c} P^2_\textrm{c} |\Xi^{(\textrm{c})}_{b k}|^2 \|x^{(\textrm{c})}_{b k}\|^{-2\alpha_{\textrm{c}}}}]$, and 
\begin{align}
\mathbb{E} [ e^{-z \hat{I}^{(\textrm{d}\to \textrm{c})}_k)}] 
=& \exp \left( -\pi \lambda \Gamma(1 - {1}/{\alpha_{\textrm{c}}}) P^{2/\alpha_{\textrm{c}}}_\textrm{d} \mathbb{E}[ \Xi^{2/\alpha_{\textrm{c}}} ] z^{1/\alpha_{\textrm{c}}} \right).
\label{eq:51}
\end{align}
\end{pro}
\begin{proof}
For any $x>0$, $\log(1+x) = \int_0^\infty \frac{1}{z} (1-e^{-xz}) e^{-z} \dint z$ \cite{hamdi2008capacity}. Therefore,
\begin{align}
\mathbb{E} \left[ \log\left(1 + \frac{X}{Y} \right) \right] = \mathbb{E} \left[\int_0^\infty \frac{1}{z} (1-e^{-z \frac{X}{Y}}) e^{-z} \dint z \right] = \mathbb{E} \left[\int_0^\infty \frac{1}{z} (1-e^{-z X}) e^{-zY} \dint z \right].
\end{align}
Using the above equality, the linearity of expectation, and the independence of $\hat{S}^{(\textrm{c})}_k$, $\hat{I}^{(\textrm{c}\to \textrm{c})}_k$ and $\hat{I}^{(\textrm{d}\to \textrm{c})}_k$, 
\begin{align}
\hat{R}^{(\textrm{c})}_{k} &= \mathbb{E} \left[\int_0^\infty \frac{1}{z} (1-e^{-z \hat{S}^{(\textrm{c})}_k}) e^{-z(\hat{I}^{(\textrm{c}\to \textrm{c})}_k + \hat{I}^{(\textrm{d}\to \textrm{c})}_k)} \dint z \right] \notag \\
&= \int_0^\infty \frac{1}{z} (1- \mathbb{E} [e^{-z \hat{S}^{(\textrm{c})}_k}]) \mathbb{E} [e^{-z\hat{I}^{(\textrm{c}\to \textrm{c})}_k}] \mathbb{E} [ e^{-z \hat{I}^{(\textrm{d}\to \textrm{c})}_k)}] \dint z .
\end{align}
The expressions for $\mathbb{E} [e^{-z \hat{S}^{(\textrm{c})}_k}]$ and $\mathbb{E} [e^{-z\hat{I}^{(\textrm{c}\to \textrm{c})}_k}]$ follow by definitions. Using the Laplace functional of the PPP $\Phi$ \cite{baccelli2003elements}, we have 
\begin{align}
\mathbb{E} [ e^{-z \hat{I}^{(\textrm{d}\to \textrm{c})}_k)}] 
=& \exp\left( - 2\pi \lambda \int_0^\infty \left( 1- \mathbb{E} [\exp(-z P^2_\textrm{d} \Xi^2 r^{-2\alpha_{\textrm{c}}})] \right) r \dint r  \right),
\end{align}
which equals (\ref{eq:51}).
\end{proof}

\section{Simulation and Numerical Results}
\label{sec:sim}

In this section, we provide simulation and numerical results to demonstrate the analytical results and obtain insights into how the various system parameters affect the cellular and D2D spectral efficiencies. The specific parameters used are summarized in Table \ref{tab:sys:para} unless otherwise specified. The cellular network consists of $19$ hexagonal cells; the side length of each cell is $R_c$. There are $K$ uniformly distributed cellular UEs in each cell, while D2D UEs are distributed as a PPP. The shadowing is lognormal with deviation $\sigma$ (dB). The pathloss parameters given in Table \ref{tab:sys:para} correspond to a carrier frequency of $2$ GHz. Specifically, we use the 3GPP macrocell propagation model (urban area) for UE-BS channels  \cite{3gppRF} and the revised \textit{Winner + B1} model (non-light-of-sight with $-5$ dB offset) for UE-UE channels \cite{3gppD2dRF}. Note that different pathloss reference values $C_{\textrm{c},0}$ and $C_{\textrm{d},0}$ are used in the UE-BS and UE-UE channels. Therefore, when evaluating the analytical expressions using the parameters in Table \ref{tab:sys:para}, $P_\textrm{c} = 23- C_{\textrm{c},0}$ (dBm) and $P_\textrm{d} = 13- C_{\textrm{c},0}$ (dBm) for the UE-BS channels while $P_\textrm{c} = 23- C_{\textrm{d},0}$ (dBm) and $P_\textrm{d} = 13- C_{\textrm{d},0}$ (dBm) for the UE-UE channels.


%
%
%

\begin{table}
\centering
\begin{tabular}{|l||r|} \hline
BS coverage radius $R_c$  & $500$ m  \\ \hline 
$\#$ cellular UEs $K$  & $4$  \\ \hline 
Density of D2D UEs $\lambda$   & $\frac{12}{\pi R_c^2}$ m$^{-2}$  \\ \hline 
$\#$ BS antennas $M$   & $100$  \\ \hline 
$\#$ UE Rx antennas $N$   & $4$  \\ \hline 
UE-BS PL exponent  $\alpha_\textrm{c}$ & $3.76$ \\ \hline
UE-UE PL exponent $\alpha_\textrm{d}$ & $4.37$ \\ \hline
UE-BS PL reference $C_{\textrm{c},0}$ & $15.3$ dB \\ \hline
UE-UE PL reference $C_{\textrm{d},0}$ & $38.5$ dB \\ \hline
Cellular Tx power $P_\textrm{c}$  & $23$ dBm  \\ \hline
D2D Tx power $P_\textrm{d}$  & $13$ dBm  \\ \hline
Channel bandwidth    & $10$ MHz \\ \hline
Noise PSD    & $-174$ dBm/Hz \\ \hline
BS noise figure    & $6$ dB \\ \hline
UE noise figure    & $9$ dB \\ \hline
Lognormal shadowing $\sigma$   & $7$ dB \\ \hline
\end{tabular}
\caption{Simulation/Numerical Parameters}
\label{tab:sys:para}
\end{table}

%
%

We first compare the simulated cellular spectral efficiency to the corresponding analytical lower bound (\ref{eq:8}) under various PZF parameters $(m_\textrm{c}, m_\textrm{d})$ in Fig. \ref{fig:2}. The conditioned random variables in (\ref{eq:8}) are averaged out in Fig. \ref{fig:2}.
We can see that the analytical lower bound (\ref{eq:8}) closely matches the simulation. The larger $m_\textrm{d}$, the better match between the simulation and the analytical lower bound (\ref{eq:8}). This is because larger $m_\textrm{d}$ implies less D2D-to-cellular interferers and thus smaller interference variance. As a result, the lower bound based on Jensen's inequality becomes more accurate with larger $m_\textrm{d}$. 
Comparing the spectral efficiency with $(m_\textrm{c}, m_\textrm{d}) = (0,2)$ to that of $(m_\textrm{c}, m_\textrm{d}) = (3,2)$, we can see that the latter is better and the spectral efficiency gain is about $1.6$ bps/Hz. This implies that it is beneficial to appropriately suppress the cochannel cellular interference in practical non-asymptotic regime. 

\begin{figure}
\centering
\includegraphics[width=10cm]{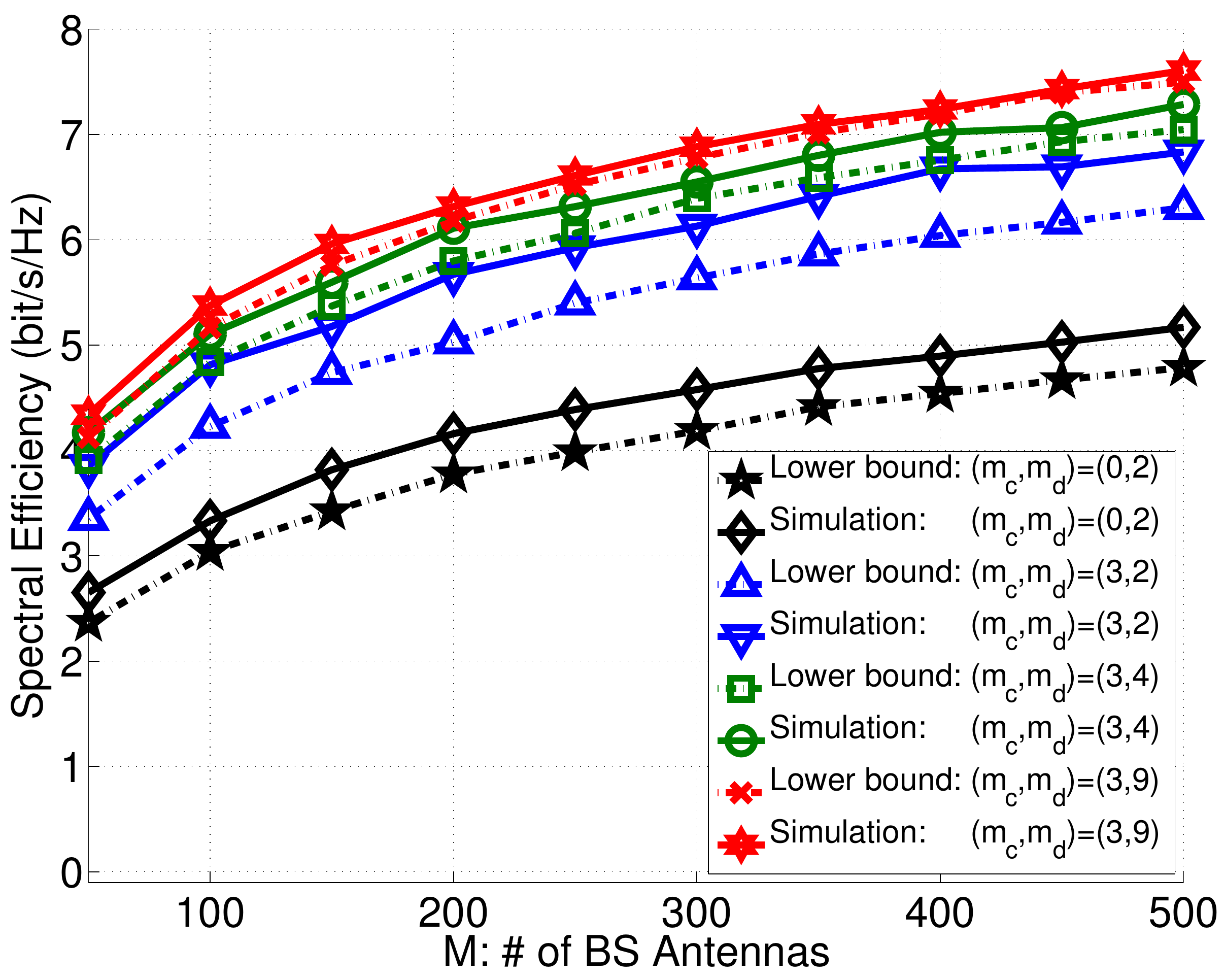}
\caption{Simulated cellular spectral efficiency vs. analytical lower bound (\ref{eq:8}) with perfect CSI.}
\label{fig:2}
\end{figure}

Since the lower bound (\ref{eq:8}) is accurate, next we use it to demonstrate the cellular spectral efficiency with scaled cellular transmit power (i.e., $P_\textrm{c} \to P_\textrm{c}/M$) in Fig. \ref{fig:4}. We consider two PZF choices: PZF with constant $m_\textrm{d}$  and PZF with scaled $m_\textrm{d}=\Theta(\sqrt{M})$. As a benchmark, we also include the curves corresponding the scenarios without D2D underlay. Also, D2D transmit power is decreased by $10$ times to accelerate the convergence.
Several observations from Fig. \ref{fig:4} are in order. First, unlike the case with unscaled cellular transmit power, Fig. \ref{fig:4} shows that it does not matter asymptotically whether cellular interference is canceled or not. Second, adopting a constant $m_\textrm{d}$ results in a fixed loss in the cellular spectral efficiency due to the underlaid D2D interference; this loss cannot be overcome by increasing the number of BS antennas when the cellular transmit power is also scaled down as $\Theta(1/M)$. This observation confirms the analytical results in Prop. \ref{pro:4}. Third, the loss in the cellular spectral efficiency due to D2D underlay can be overcome by scaling $m_\textrm{d}$ as $\Theta(\sqrt{M})$, validating the theoretical finding in Prop. \ref{pro:5}. But the convergence rate is relatively slow.

\begin{figure}
\centering
\includegraphics[width=10cm]{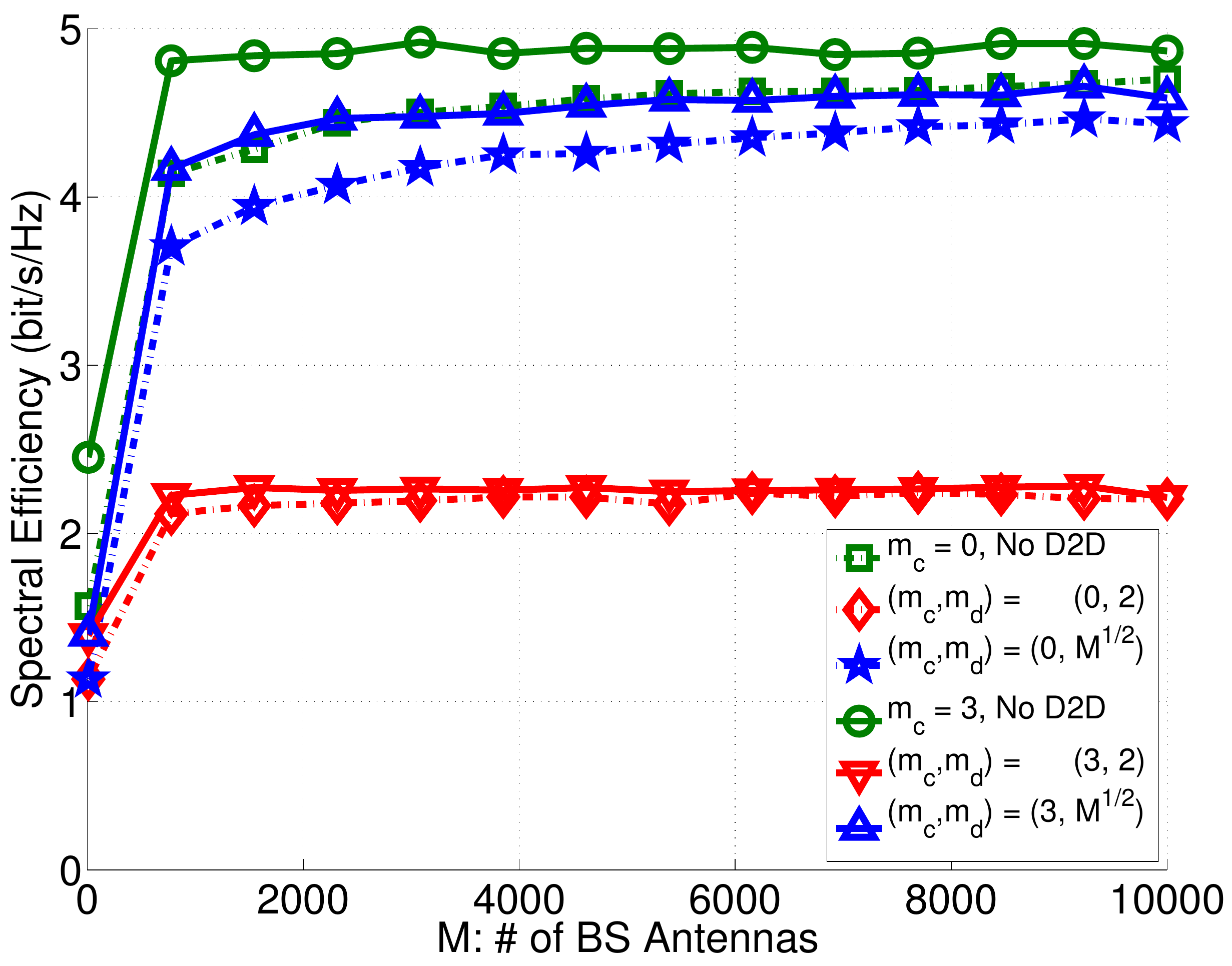}
\caption{Cellular spectral efficiency with scaled cellular transmit power and perfect CSI.}
\label{fig:4}
\end{figure}


Fig. \ref{fig:3} compares the simulated D2D spectral efficiency to the corresponding analytical lower bound (\ref{eq:12}) under different deterministic D2D distances and $(n_\textrm{c},n_\textrm{d})=(0,2)$.  The conditioned random variables in (\ref{eq:12}) are averaged out in Fig. \ref{fig:3}. We can see that the analytical lower bound (\ref{eq:12}) closely matches the simulation when $N\geq 6$ while being a bit loose when $N<6$. The accuracy of the lower bound obtained from Jensen's inequality implies that after canceling $2$ nearest D2D interferers, the variance of the residual interference is relatively small. Fig. \ref{fig:3} also shows that D2D spectral efficiency is quite sensitive to its communication range: there is a loss of about $3$ bps/Hz in spectral efficiency if D2D range is increased from $20$ m to $35$ m. 


\begin{figure}
\centering
\includegraphics[width=10cm]{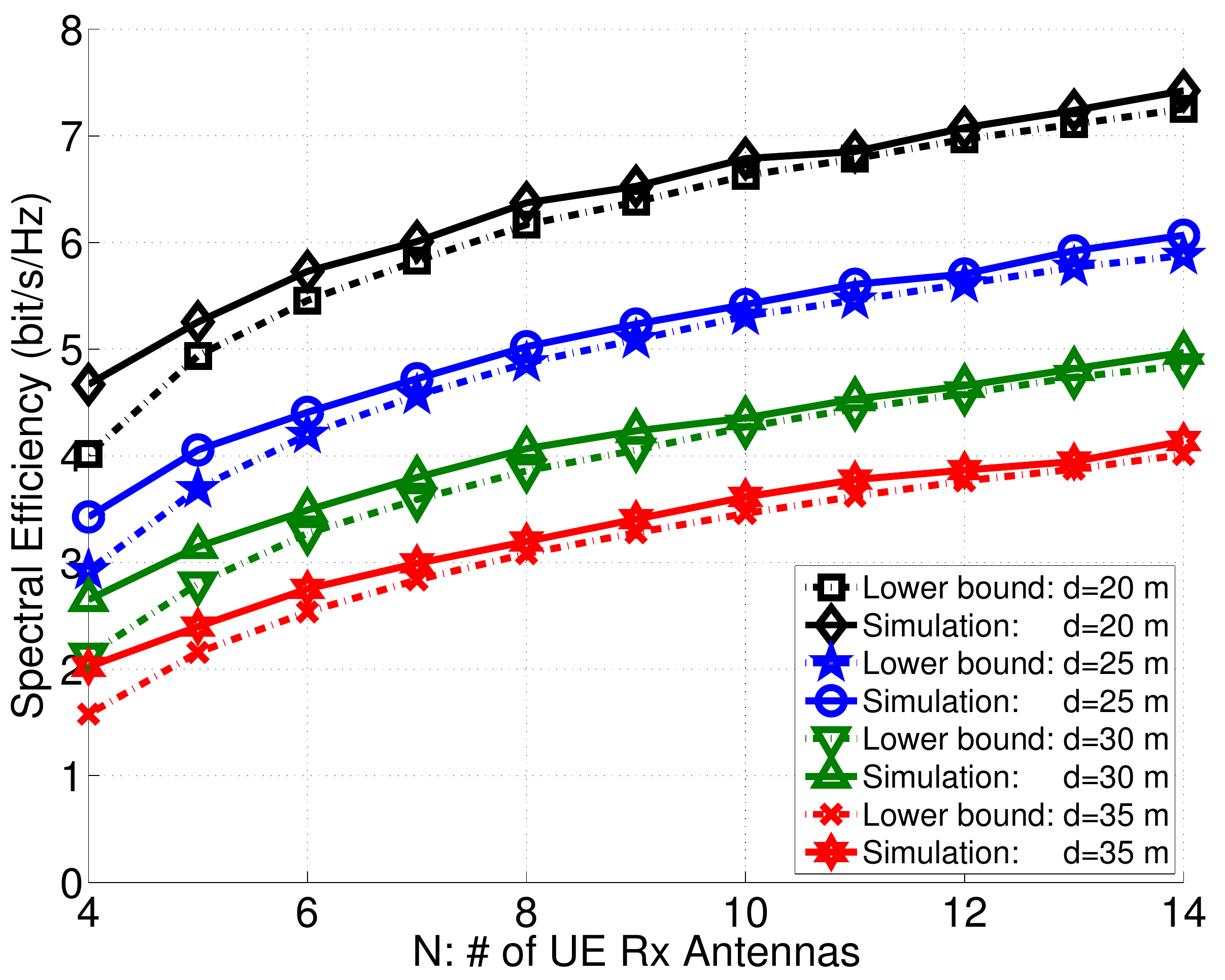}
\caption{Simulated D2D spectral efficiency vs. analytical lower bound (\ref{eq:12}) with perfect CSI and $(n_\textrm{c},n_\textrm{d})=(0,2)$.}
\label{fig:3}
\end{figure}

Next we evaluate the effect of multi-user cellular transmission on D2D spectral efficiency. Fig. \ref{fig:5} shows
the D2D spectral efficiency as a function of the number $K$ of co-channel cellular UEs per cell. Not surprisingly, as $K$ increases, D2D spectral efficiency decreases due to the increased cellular-to-D2D interference. The interesting observation from Fig. \ref{fig:5} is that even with $(n_\textrm{c},n_\textrm{d})=(0,0)$ (i.e., the MRC receiver) the average D2D spectral efficiency is not severely affected by scaling up the number of cellular UEs. For example, when $K$ increases from $10$ to $20$, the loss in D2D spectral efficiency is less than $0.5$ bps/Hz. This implies that we can scale up the uplink capacity in a massive MIMO system without much loss in the average D2D spectral efficiency. 

\begin{figure}
\centering
\includegraphics[width=10cm]{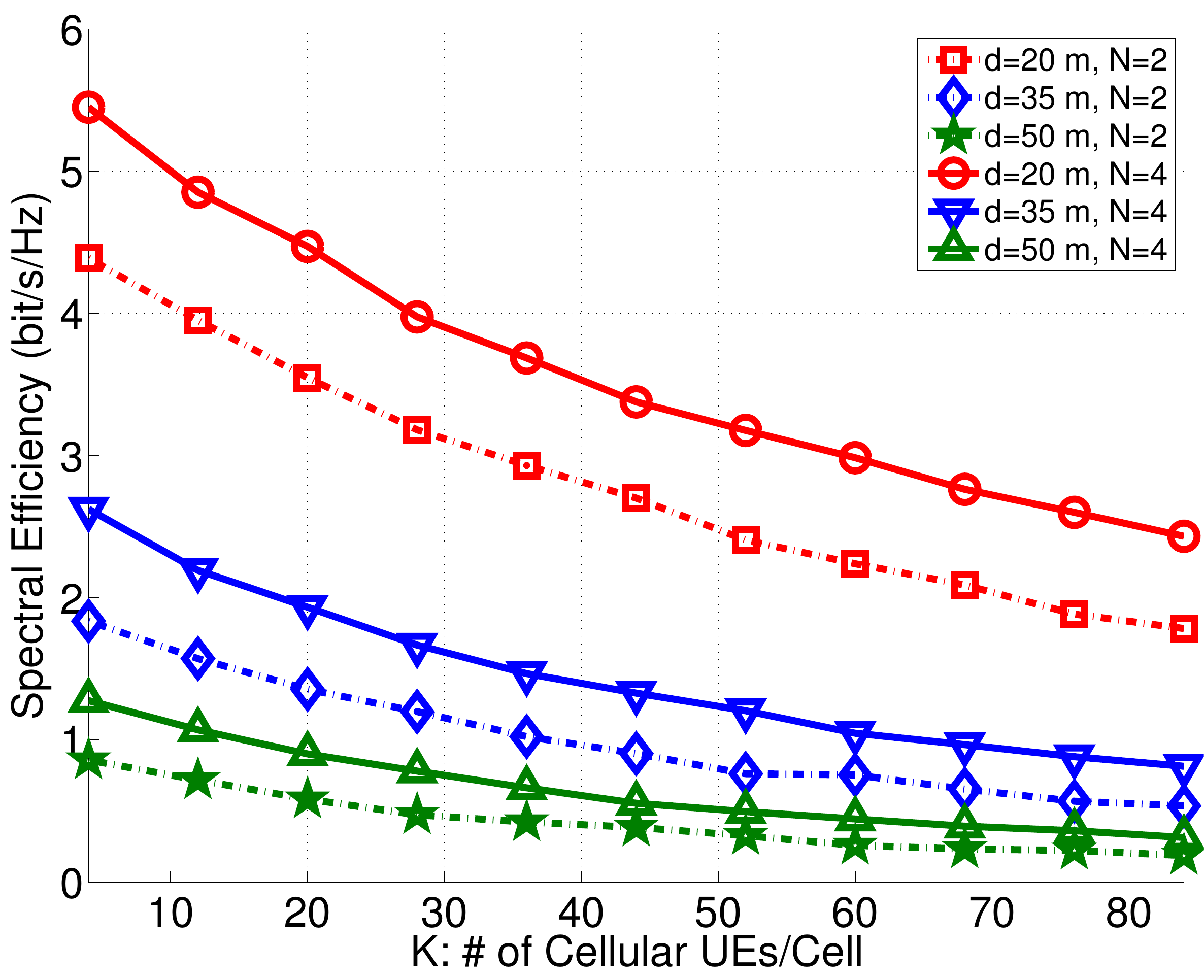}
\caption{Effect of multi-user cellular transmission on D2D spectral efficiency with perfect CSI and $(n_\textrm{c},n_\textrm{d})=(0,0)$.}
\label{fig:5}
\end{figure}

%
%

%
%

Fig. \ref{fig:8} illustrates how PZF parameters $(m_\textrm{c},m_\textrm{d})$ should be chosen to optimize cellular spectral efficiency. Several observations are in order. First, it is beneficial to suppress intra-cell cellular interference since the spectral efficiency gain is large as $m_\textrm{c}$ increases from $0$ to $3$. Second, canceling further other-cell cellular interference (i.e., $m_\textrm{c}=4$) provides additional marginal gain; this gain may not be large enough to justify the additional training and BS coordination overhead. Third, it suffices to cancel about 2 to 4 D2D interferers to get close-to-optimal performance. Of course, the last observation depends on D2D transmitter density $\lambda$ and the number of BS antennas $M$. With larger $\lambda$ and $M$, we may like to cancel a few more D2D interferers.
For D2D spectral efficiency, as pointed out in Section \ref{sec:model}, a general PZF filter is hard to implement in practice; instead, a simple MRC filter with $(n_\textrm{c},n_\textrm{d})=(0,0)$ or a MMSE filter should be used. Therefore, we do not consider optimizing D2D spectral efficiency over PZF parameters.

\begin{figure}
\centering
\includegraphics[width=10cm]{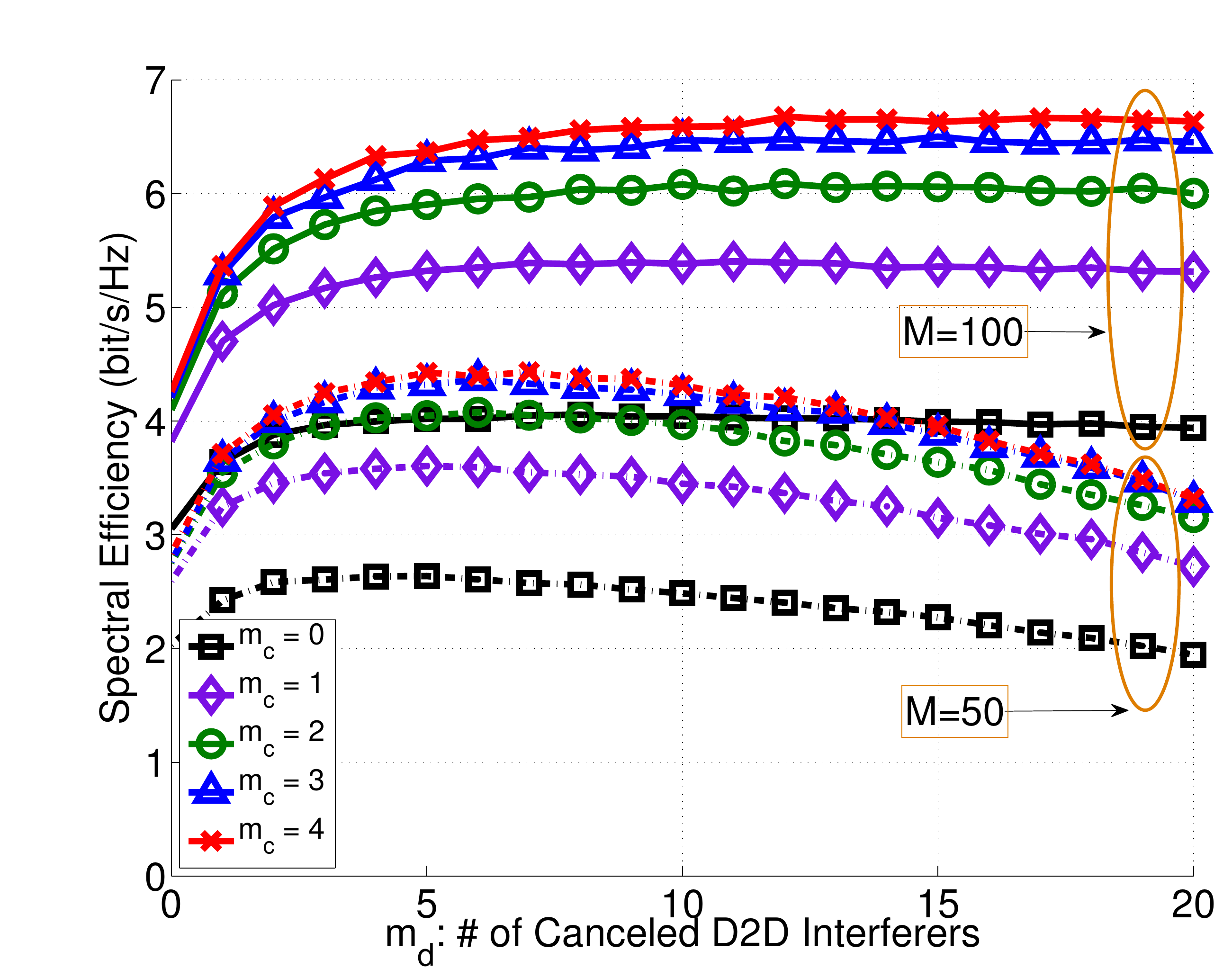}
\caption{Optimizing cellular spectral efficiency over PZF parameters.}
\label{fig:8}
\end{figure}

Fig. \ref{fig:7} illustrates the effect of D2D underlay contamination on the asymptotic cellular spectral efficiency of massive MIMO. Compared to the case without D2D, where only pilot contamination exists, D2D underlay contamination degrades the achievable asymptotic massive MIMO spectral efficiency. For example, with shadowing deviation $\sigma=7$ dB and $\pi R^2_c \lambda=4$, the spectral efficiency is reduced from $6$ bps/Hz to about $3.8$ bps/Hz. Further, the more the underlaid D2D UEs, the smaller the asymptotic cellular spectral efficiency. Fig. \ref{fig:7}  shows that when $\pi R^2_c \lambda\geq 22$ the effect of D2D underlay contamination dominates in the overall effect of pilot and underlay contamination.

\begin{figure}
\centering
\includegraphics[width=10cm]{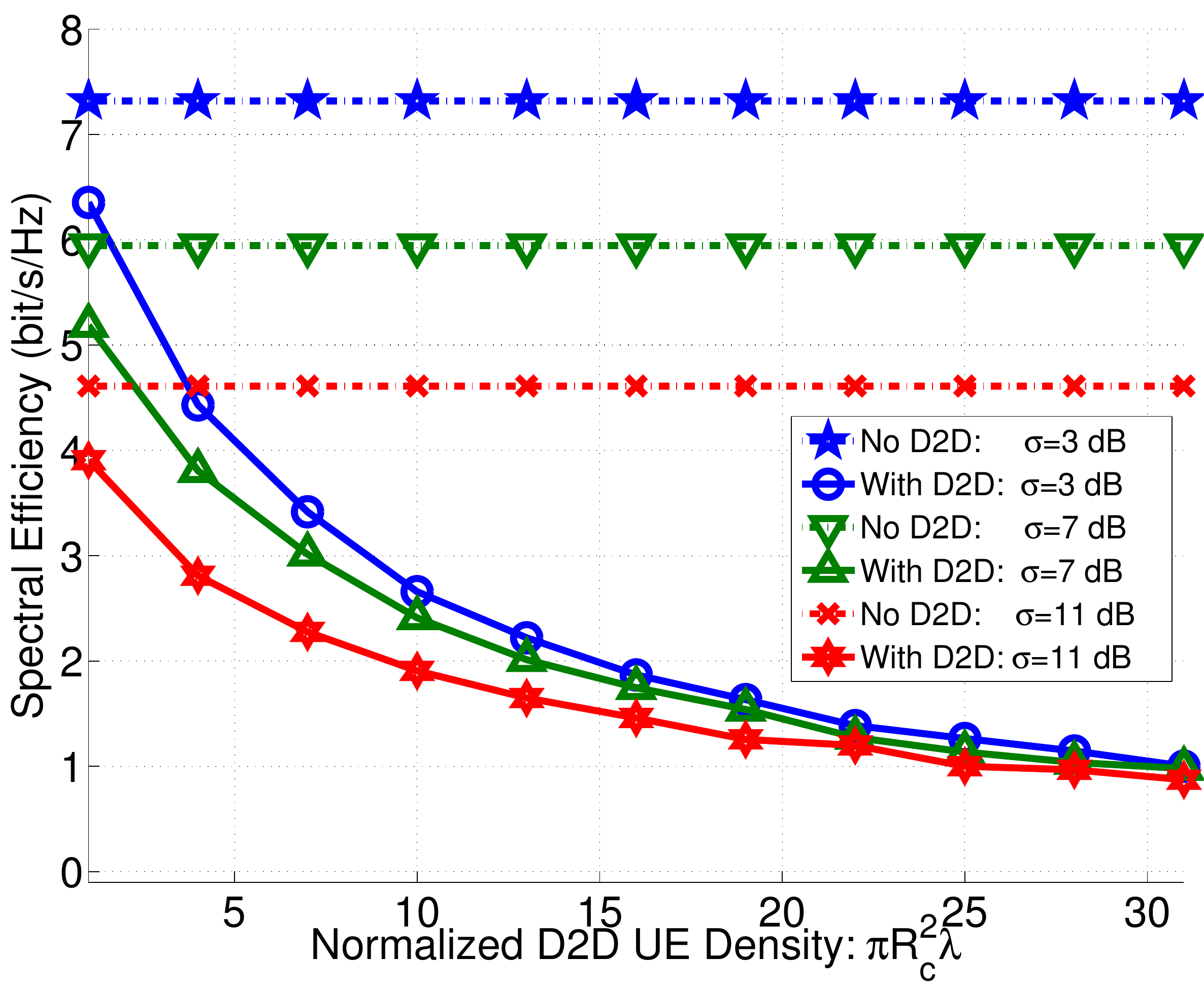}
\caption{Effect of D2D underlay contamination on asymptotic cellular spectral efficiency of massive MIMO with $(m_\textrm{c},m_\textrm{d})=(0,0)$ and $T_\textrm{c}=4$.}
\label{fig:7}
\end{figure}

\section{Conclusions}
\label{sec:conclusions}

In this paper, we have studied the spectral efficiency of a D2D underlaid massive MIMO system under perfect and imperfect CSI. We have found that massive MIMO can efficiently handle the D2D-to-cellular interference. Meanwhile, from an average perspective, D2D links are relatively robust to the cellular-to-D2D interference even if there are quite many cochannel cellular users. D2D interference does make the estimated CSI in massive MIMO less accurate and thus in turn hurts the cellular spectral efficiency. One simple approach to alleviating this effect is to deactivate D2D links in the cellular training phase. Overall, our study suggests that D2D may be much simpler in massive MIMO cellular systems than in current cellular systems.

Spectral efficiency is used as the sole metric throughout this paper. Future work may carry out throughput analysis by taking into account the overhead cost in channel training, scheduling and possibly retransmissions. Also, it is of interest to consider other more sophisticated receivers like MMSE receivers and linear receivers with successive interference cancellation and compare their system-level performance with that of PZF receivers studied in this paper.  

\section*{Acknowledgment}

The authors thank Editor Chia-Han Lee and the anonymous reviewers for their valuable comments and suggestions. The authors also thank Yingxiao Zhang for bringing Lemma 1 in \cite{hamdi2008capacity} to our attention.

\appendix

\subsection{Proof of Proposition \ref{pro:1}}
\label{proof:pro:1}

We show that a PZF receiver with $m_\textrm{c} = m_\textrm{d} =0$, i.e., the MRC receiver, at the BS suffices. With $m_\textrm{c}= m_\textrm{d} =0$, the PZF receiver $\bw^{(\textrm{c})}_{k} = \bh^{(\textrm{c})}_{0k}$. 
By the law of large numbers, $\frac{1}{M} \| \bh_{0k}^{(\textrm{c})} \|^2    \as 1$, $\frac{1}{M}\bh^{(\textrm{c})*}_{0k}  \bh_{b\ell}^{(\textrm{c})} \as 0, \ell \neq k$ or $b\neq 0$. It follows that when conditioned on $\Xi^{(\textrm{c})}_{0k}$ and $x^{(\textrm{c})}_{0k}$,
\begin{align}
\frac{1}{M^2} P_\textrm{c} \Xi^{(\textrm{c})}_{0k} \|x^{(\textrm{c})}_{0k}\|^{-\alpha_\textrm{c}} \| \bh^{(\textrm{c})}_{0k} \|^4 \as P_\textrm{c} \Xi^{(\textrm{c})}_{0k} \|x^{(\textrm{c})}_{0k}\|^{-\alpha_\textrm{c}}.
\end{align}
Also, the noise term normalized by $M^2$ converges as $\frac{1}{M^2}  N_0  \| \bh_{0k}^{(\textrm{c})} \|^2 \as 0$.
Further, interchanging the order of the limit and the finite sum,  the cellular interference normalized by $M^2$ converges as
\begin{align}
&\lim_{M \to \infty} \frac{1}{M^2} \sum_{b=0}^B \sum_{\ell \in \mathcal{K}^{(\textrm{c})}_{bk}  } P_\textrm{c} \Xi_{b\ell}^{(\textrm{c})} \|x_{b\ell}^{(\textrm{c})} \|^{-\alpha_\textrm{c}} |\bh_{0k}^{(\textrm{c})*} \bh_{b\ell}^{(\textrm{c})} |^2  \notag \\
&=  \sum_{b=0}^B \sum_{\ell \in \mathcal{K}^{(\textrm{c})}_{bk}  } P_\textrm{c} \Xi_{b\ell}^{(\textrm{c})} \|x_{b\ell}^{(\textrm{c})} \|^{-\alpha_\textrm{c}} \left(\lim_{M \to \infty}   \frac{1}{M^2} |\bh_{0k}^{(\textrm{c})*} \bh_{b\ell}^{(\textrm{c})} |^2  \right)  \as 0 . 
\end{align}

Next we show that the D2D interference normalized by $M^2$ converges to $0$ as $M \to \infty$. Note that in this case \textit{we cannot directly interchange the order of the limit and the infinite sum} to conclude that it converges to $0$ almost surely. Instead, we can prove its convergence in probability, i.e., for any $\epsilon > 0$,
\begin{align}
\lim_{M \to \infty} \mathbb{P} \left(   \frac{1}{M^2}\sum_{i \in \Phi}  P_\textrm{d} \Xi_i^{(\textrm{d})} \|x_i^{(\textrm{d})}\|^{-\alpha_\textrm{c}} |\bh^{(\textrm{c})*}_{0k}  \bh_i^{(\textrm{d})} |^2    < \epsilon  \right) = 1 .
\label{eq:app:05}
\end{align}
To this end, we partition the D2D transmitters into two groups: one group is composed of those transmitters located within distance $r_o$ from the BS and the other group is composed of those transmitters located with distance grater than $r_o$ from the BS.  Then 
we have
\begin{align}
&\mathbb{P} \left(   \frac{1}{M^2} \sum_{i \in \Phi}  P_\textrm{d} \Xi_i^{(\textrm{d})} \|x_i^{(\textrm{d})}\|^{-\alpha_\textrm{c}} |\bh^{(\textrm{c})*}_{0k}  \bh_i^{(\textrm{d})} |^2   \geq \epsilon  \right)   \leq  \mathbb{P} \left(   X   \geq \frac{\epsilon}{2}  \right)   +\mathbb{P} \left(   Y   \geq \frac{\epsilon}{2}  \right)  .
\label{eq:4}
\end{align} 
where 
\begin{align}
X&= \frac{1}{M^2} \sum_{i \in \Phi \cap \mathcal{B}^c (o,r_o) }  P_\textrm{d} \Xi_i^{(\textrm{d})} \|x_i^{(\textrm{d})}\|^{-\alpha_\textrm{c}} |\bh^{(\textrm{c})*}_{0k}  \bh_i^{(\textrm{d})} |^2   \\
Y&= \frac{1}{M^2} \sum_{i \in \Phi \cap \mathcal{B} (o,r_o) }  P_\textrm{d}  \Xi_i^{(\textrm{d})}  \|x_i^{(\textrm{d})}\|^{-\alpha_\textrm{c}} |\bh^{(\textrm{c})*}_{0k}  \bh_i^{(\textrm{d})} |^2 .
\end{align}
Next we show in two steps that the two terms on the right hand side of (\ref{eq:4}) can be made arbitrarily small by choosing $M$ large enough.

\textit{Step 1.} For the first term on the right hand side of (\ref{eq:4}), we have
\begin{align}
\mathbb{P} \left(  X   
\geq \frac{\epsilon}{2}  \right)  &\leq \frac{2}{\epsilon} \mathbb{E} \left[  X \right] =  \frac{2P_\textrm{d}\bar{\Xi}}{\epsilon M}   \mathbb{E} [   \sum_{i \in \Phi \cap \mathcal{B}^c (o,r_o) }   \|x_i^{(\textrm{d})}\|^{-\alpha_\textrm{c}}   ]  \label{eq:app:02} \\
&=  \frac{4\pi \lambda P_\textrm{d}\bar{\Xi}}{\epsilon M}   \int_{r_o}^{\infty}  r^{1-\alpha_\textrm{c}}  \dint   r =  \frac{4\pi \lambda P_\textrm{d}\bar{\Xi}}{\epsilon M}  \frac{1}{(\alpha_\textrm{c} -2)r_o^{\alpha_\textrm{c} -2}}, \label{eq:app:04} 
\end{align} 
where the first inequality is due to the Markov inequality, the first equality is due to $\mathbb{E}[\Xi_i^{(\textrm{d})} ] = \bar{\Xi}$ and $\mathbb{E}[|\bh_{0k}^{(\textrm{c})*}  \bh_i^{(\textrm{d})} |^2 ] = M$, and the second equality is due to Campbell's formula \cite{baccelli2003elements}, and we use the assumption that $\alpha_\textrm{c} >2$ in the last equality.
It follows that that there exists $M_1$ large enough such that for all $M \geq M_1$,
\begin{align}
\mathbb{P} \left(  \frac{1}{M^2} \!\!\!\! \sum_{i \in \Phi \cap \mathcal{B}^c (o,r_o) } \!\!\!\!  P_\textrm{d} \Xi_i^{(\textrm{d})} \|x_i^{(\textrm{d})}\|^{-\alpha_\textrm{c}} |\bh^{(\textrm{c})*}_{0k}  \bh_i^{(\textrm{d})} |^2    \geq \frac{\epsilon}{2}  \right) < \frac{\delta}{2},
\label{eq:5}
\end{align}
where $\delta$ is an arbitrary small positive constant.

\textit{Step 2.} For the second term on the right hand side of (\ref{eq:4}), 
\begin{align}
& \mathbb{P} \left(  Y   \geq \frac{\epsilon}{2}  \right)  = \mathbb{P} \left(  Y   \geq \frac{\epsilon}{2} \big|  E \right)  \mathbb{P} \left(   E \right)   + \mathbb{P} \left(  Y       \geq \frac{\epsilon}{2} \big|  E^c \right)  \mathbb{P} \left( E^c \right) . \notag 
\end{align} 
Here $E = \{|\Phi \cap \mathcal{B} (o,r_o)| \leq C\}$, where $C$ is a constant to be chosen, and $E^c$ is the complement of $E$.

\textit{Step 2(a).} Note that the number of D2D transmitters in $\mathcal{B} (o,r_o)$, denoted as $ |\Phi \cap \mathcal{B} (o,r_o)|$, is Poisson distributed with mean $\lambda \pi r_o^2$. We can choose $C$ large enough but finite such that
\begin{align}
& \mathbb{P} \left(  Y       \geq \frac{\epsilon}{2} \big|  E^c \right)  \mathbb{P} \left( E^c \right)  \leq \mathbb{P} \left(  E^c \right)   = 1 - \sum_{n=0}^{C} \frac{ (\lambda \pi r_o^2)^n }{n!} e^{ - \lambda \pi r_o^2}   < \frac{\delta}{4} .
\label{eq:6}
\end{align}
Note that the choice of $C$ depends on $\delta$, and thus we write $C(\delta)$ to explicitly spell out this dependency.

\textit{Step 2(b).} Since $\frac{1}{M}\bh^{(\textrm{c})*}_{0k}  \bh_i^{(\textrm{d})} \as 0$, conditioning on $|\Phi \cap \mathcal{B} (o,r_o)| \leq C(\delta)$, we have
$
Y   \as 0.
$
It follows that there exists $M_2$ large enough such that for all $M \geq M_2$,
\begin{align}
&\mathbb{P} \left(  Y     \geq \frac{\epsilon}{2} \big|  E \right)  \mathbb{P} \left(   E \right)   \leq \mathbb{P} \left(  Y  \geq \frac{\epsilon}{2} \big|  E \right) < \frac{\delta}{4} .
\label{eq:7}
\end{align}
 
Combining (\ref{eq:5}), (\ref{eq:6}) and (\ref{eq:7}) obtained in Steps 1, 2(a) and 2(b) respectively, we can find $C(\delta)$ large enough such that for all $M \geq \max \{M_1,M_2\}$,
$$
\mathbb{P} \left(    \frac{1}{M^2} \sum_{i \in \Phi }  P_\textrm{d} \Xi_i^{(\textrm{d})} \|x_i^{(\textrm{d})}\|^{-\alpha_\textrm{c}} |\bh^{(\textrm{c})*}_{0k}  \bh_i^{(\textrm{d})} |^2    \geq \epsilon  \right) \leq \frac{\delta}{2} + \frac{\delta}{4} + \frac{\delta}{4} =  \delta .
$$
As $\delta$ is an arbitrary positive constant, we conclude that (\ref{eq:app:05}) holds. This completes the proof.

\subsection{Proof of Proposition \ref{pro:4}}
\label{proof:pro:4}

When the transmit powers of cellular UEs scale as $P_\textrm{c}/M$, as in the proof of Prop. \ref{pro:1}, we can show that as $M \to \infty$, the desired signal power $S^{(\textrm{c})}_k$, the cellular interference power $I^{(\textrm{c}\to \textrm{c})}_k$,  and the noise power $ \| \bw_k^{(\textrm{c})}  \|^2 N_0$ normalized by $M$ converge as follows.  
\begin{align}
&\lim_{M \to \infty} \frac{1}{M}  S^{(\textrm{c})}_k \as  P_\textrm{c} \Xi_{0k}^{(\textrm{c})} \|x_{0k}^{(\textrm{c})}\|^{-\alpha_\textrm{c}},  
\lim_{M \to \infty} \frac{1}{M}  I^{(\textrm{c}\to \textrm{c})}_k \as  0,   \lim_{M \to \infty} \frac{1}{M} \| \bw_k^{(\textrm{c})}  \|^2 N_0 \as  N_0 .
\end{align}

Next we study how the D2D interference power behaves asymptotically. To this end, we first prove the following lemma.
\begin{lem}
Consider a PPP $\{x_i\}$ in $\mathbb{R}^2$ with density $\lambda$. Each point $x_i$ is associated with a sequence of non-negative random marks $\{F_{i,m}\}_m$. For each $m$, $\{F_{i,m}\}_i$ are i.i.d., and $\mathbb{E}[F_{i,m}]<\infty$. Further, $F_{i,m} \cd F_{i,\infty}$. If $\mathbb{E}[F_{i,m}^{\frac{2}{\alpha}}] \to \mathbb{E}[F_{i,\infty}^{\frac{2}{\alpha}}]$, where  $\alpha > 2$ is a constant, then $Y_m \triangleq \sum_i \| x_i \|^{-\alpha} F_{i,m}$ is well defined and $Y_m \cd Y_\infty \triangleq \sum_i \| x_i \|^{-\alpha}  F_{i,\infty}$.
\label{lem:3}
\end{lem}
\begin{proof}
By definition, $Y_m$ equals the value of a shot-noise random field evaluated at the origin. The shot-noise random field is associated with a marked PPP. As $\mathbb{E}[F_{i,m}]<\infty$, $Y_m$ is almost surely finite \cite{baccelli2009stochastic} and thus is well defined. To show  $Y_m \cd Y$, we show that the Laplace transform of the former converges to that of the latter as follows.
\begin{align}
\lim_{m\to\infty} L_{Y_m} (s) &=\lim_{m\to\infty}   \mathbb{E} [ e^{-sY_m}] 
= \lim_{m\to\infty} \exp \left( -C(\alpha) s^{\frac{2}{\alpha}} \mathbb{E}[ F_{1,m}^{\frac{2}{\alpha}} ] \right) \notag \\
&= \exp \left( -C(\alpha) s^{\frac{2}{\alpha}} \lim_{m\to\infty} \mathbb{E}[ F_{1,m}^{\frac{2}{\alpha}} ] \right) = \exp \left( -C(\alpha) s^{\frac{2}{\alpha}}  \mathbb{E}[ F_{1,\infty}^{\frac{2}{\alpha}} ] \right) = L_{Y_\infty} (s) ,
\end{align}
where $C(\alpha) = \pi \lambda \Gamma(1 - {1}/{\alpha})$, we have used the Laplace functional of the PPP $\Phi$ \cite{baccelli2009stochastic} in the second equality and the assumption $\mathbb{E}[F_{i,m}^{\frac{2}{\alpha}}] \to \mathbb{E}[F_{i,\infty}^{\frac{2}{\alpha}}]$ in the penultimate equality.
\end{proof}
To apply Lemma \ref{lem:3}, with a slight abuse of notation, we also denote by $I^{(\textrm{d}\to \textrm{c})}_k$ the asymptotic interference power $\sum_{i \in \Phi_k^{(\textrm{c})}}    \|x_i^{(\textrm{d})}\|^{-\alpha_\textrm{c}} F_{i,\infty} $ and 
$
I^{(\textrm{d}\to \textrm{c})}_k (M) = \sum_{i \in \Phi_k^{(\textrm{c})}}  \|x_i^{(\textrm{d})}\|^{-\alpha_\textrm{c}}  F_{i,M}  ,
$
where $F_{i,\infty} = P_\textrm{d} \Xi_i^{(\textrm{d})}  \eta_i$ and $F_{i,M} = P_\textrm{d} \Xi_i^{(\textrm{d})}   \frac{1}{M} |\bw^{(\textrm{c})*}_k  \bh^{(\textrm{d})}_i|^2 $. 
By the Central Limit Theorem, $\frac{1}{\sqrt{M}}  \bw^{(\textrm{c})*}_k  \bh^{(\textrm{d})}_i \cd \mathcal{CN}(0,1)$, where $\cd$ denotes convergence in distribution. It follows that $F_{i,M} \cd F_{i,\infty}$. Note that
\begin{align}
I^{(\textrm{d}\to \textrm{c})}_k (M)  &= \sum_{i \in \Phi}    \|x_i^{(\textrm{d})}\|^{-\alpha_\textrm{c}} F_{i,M} - \sum_{i \in \Phi \backslash \Phi_k^{(\textrm{c})}}    \|x_i^{(\textrm{d})}\|^{-\alpha_\textrm{c}} F_{i,M}  \label{eq:rv2:1} \\
I^{(\textrm{d}\to \textrm{c})}_k &= \sum_{i \in \Phi}    \|x_i^{(\textrm{d})}\|^{-\alpha_\textrm{c}} F_{i,\infty} - \sum_{i \in \Phi \backslash \Phi_k^{(\textrm{c})}}    \|x_i^{(\textrm{d})}\|^{-\alpha_\textrm{c}} F_{i,\infty}. \label{eq:rv2:2} 
\end{align}
For any finite fixed $m_\textrm{d}$, the second term on the right hand side of (\ref{eq:rv2:1}) converges in distribution to the second term on the right hand side of (\ref{eq:rv2:2}). It remains to show that $\sum_{i \in \Phi}    \|x_i^{(\textrm{d})}\|^{-\alpha_\textrm{c}} F_{i,M} \cd \sum_{i \in \Phi}    \|x_i^{(\textrm{d})}\|^{-\alpha_\textrm{c}} F_{i,\infty}$, which holds if $\mathbb{E}[F_{i,M}^{\frac{2}{\alpha_\textrm{c}}}] \to \mathbb{E}[F_{i,\infty}^{\frac{2}{\alpha_\textrm{c}}}]$ by Lemma \ref{lem:3}. By Theorem 5.5.2 in \cite{durrett2010probability}, $\mathbb{E}[F_{i,M}^{\frac{2}{\alpha_\textrm{c}}}] \to \mathbb{E}[F_{i,\infty}^{\frac{2}{\alpha_\textrm{c}}}]$ if and only if $\{F_{i,M}^{\frac{2}{\alpha_\textrm{c}}}\}$ are uniformly integrable. A sufficient condition for $\{X_m\}_m$ to be uniformly integrable is that $\mathbb{E}[|X_m|^p]\leq C<\infty, \forall m$, where $p>1$ \cite{durrett2010probability}.
In our case,
\begin{align}
\mathbb{E}[ |F_{i,M}^{\frac{2}{\alpha_\textrm{c}}} |^{\frac{\alpha_\textrm{c}}{2}} ] = \mathbb{E} [F_{i,M}] = \mathbb{E} \left[ P_\textrm{d} \Xi_i^{(\textrm{d})}   \frac{1}{M} |\bw^{(\textrm{c})*}_k  \bh^{(\textrm{d})}_i|^2 \right] \leq P_\textrm{d} \bar{\Xi} <\infty.
\end{align}  
It follows that $\{F_{i,M}^{\frac{2}{\alpha_\textrm{c}}}\}$ are uniformly integrable and thus $\mathbb{E}[F_{i,M}^{\frac{2}{\alpha_\textrm{c}}}] \to \mathbb{E}[F_{i,\infty}^{\frac{2}{\alpha_\textrm{c}}}]$.
To sum up, the D2D-to-cellular interference converges in distribution to  $\sum_{i \in \Phi_k^{(\textrm{c})}}  P_\textrm{d} \Xi_i^{(\textrm{d})} \|x_i^{(\textrm{d})}\|^{-\alpha_\textrm{c}} \eta_i$. Therefore, the spectral efficiency of cellular UE $k$ converges as in (\ref{eq:14}).

The lower bound (\ref{eq:15}) is due to Jensen's inequality:
\begin{align}
&\mathbb{E} \left[ \log \left( 1 + \frac{  P_\textrm{c} \Xi_{0k}^{(\textrm{c})} \|x_{0k}^{(\textrm{c})}\|^{-\alpha_\textrm{c}} }{ I^{(\textrm{d}\to \textrm{c})}_k + N_0 }  \right)  \right] 
\geq \log \left( 1 + \frac{ P_\textrm{c} \Xi_{0k}^{(\textrm{c})} \|x_{0k}^{(\textrm{c})}\|^{-\alpha_\textrm{c}} }{ \mathbb{E}[ I^{(\textrm{d}\to \textrm{c})}_k ] + N_0 }  \right) .
\label{eq:app:14}
\end{align}
As the BS uses $m_\textrm{d}$ degrees of freedom to cancel the interference from the $m_\textrm{d}$ nearest D2D transmitters when detecting the signal of cellular UE $k$, $\Phi_k^{(\textrm{c})}$ consists of the points from the original PPP $\Phi$ except the nearest $m_\textrm{d}$ points to the origin. Let us order the points in $\Phi$ based on their distances to the BS in an ascending manner, i.e., $\|x_1^{(\textrm{d})} \| \leq \|x_2^{(\textrm{d})} \| \leq ... $. Then
\begin{align}
\mathbb{E} [ I^{(\textrm{d}\to \textrm{c})}_k ] &= \mathbb{E}_{\Phi} [\sum_{i \in \Phi_k^{(\textrm{c})}}  P_\textrm{d} \mathbb{E} [\Xi_i^{(\textrm{d})}] \|x_i^{(\textrm{d})}\|^{-\alpha_\textrm{c}}  \mathbb{E} [ \eta_i ] ]  = P_\textrm{d} \bar{\Xi} \mathbb{E}_{\Phi} [\sum_{i \in \Phi_k^{(\textrm{c})}}  \|x_i^{(\textrm{d})}\|^{-\alpha_\textrm{c}}   ]  \notag \\
&=P_\textrm{d} \bar{\Xi} \mathbb{E}_{\Phi} [\sum_{i=m_\textrm{d}+1}^{\infty}   \|x_i^{(\textrm{d})}\|^{-\alpha_\textrm{c}}   ].  
\label{eq:app:13}
\end{align}
Conditioning on the location  $x_{m_\textrm{d}}^{(\textrm{d})} = (r, \theta)$ of the $m_\textrm{d}$-th nearest point in $\Phi$, 
\begin{align}
\mathbb{E} [ I^{(\textrm{d}\to \textrm{c})}_k | x_{m_\textrm{d}}^{(\textrm{d})} = (r, \theta) ] 
&=P_\textrm{d} \bar{\Xi} \mathbb{E}_{\Phi} \left[ \sum_{i=m_\textrm{d}+1}^{\infty}  \|x_i^{(\textrm{d})}\|^{-\alpha_\textrm{c}}  \big| x_{m_\textrm{d}}^{(\textrm{d})} = (r, \theta)  \right]  \notag \\
&= P_\textrm{d} \bar{\Xi} 2\pi \lambda \int_{r}^{\infty} t^{1-\alpha_\textrm{c}} \dint t   =  \frac{P_\textrm{d} \bar{\Xi} 2\pi \lambda}{ \alpha_\textrm{c} - 2}   r^{2-\alpha_\textrm{c}},
\label{eq:app:06}
\end{align}
where the second equality is due to Campbell formula \cite{baccelli2003elements}. To decondition on $x_{m_\textrm{d}}^{(\textrm{d})} = (r, \theta)$, we need the PDF of $\| x_{m_\textrm{d}}^{(\textrm{d})} \|$ derived in \cite{haenggi2005distances}:
\begin{align}
f_{\|x_{m_\textrm{d}}^{(\textrm{d})} \|} (r) = \frac{ 2(\lambda \pi r^2)^{m_\textrm{d}} }{r (m_\textrm{d} - 1)!}  e^{-\lambda \pi r^2} , \quad r \geq 0.
\end{align}


Using the fact that $x_{m_\textrm{d}}^{(\textrm{d})} $ is uniform in direction and $f_{\|x_{m_\textrm{d}}^{(\textrm{d})} \|} (r)$, we decondition on $x_{m_\textrm{d}}^{(\textrm{d})} $ in (\ref{eq:app:06}) and obtain
\begin{align}
\mathbb{E} [ I^{(\textrm{d}\to \textrm{c})}_k  ] 
&=  \frac{P_\textrm{d} \bar{\Xi}2\pi \lambda}{ \alpha_\textrm{c} - 2} \int_{0}^{\infty}    r^{2-\alpha_\textrm{c}} f_{\|x_{m_\textrm{d}}^{(\textrm{d})} \|} (r) \dint r  \notag \\
&= \frac{P_\textrm{d} \bar{\Xi}2\pi \lambda}{ \alpha_\textrm{c} - 2} \cdot \frac{1}{(m_\textrm{d}-1)!} (\lambda \pi)^{\frac{\alpha_\textrm{c}}{2} - 1} \int_{0}^{\infty}    t^{m_\textrm{d}-\frac{\alpha_\textrm{c}}{2}} e^{-t} \dint t, \label{eq:app:07} 
\end{align}
where we have changed variable $t=\lambda \pi r^2$ in (\ref{eq:app:07}). By the definition of the Gamma function,
\begin{align}
\mathbb{E} [ I^{(\textrm{d}\to \textrm{c})}_k  ] 
&= \frac{2 P_\textrm{d}\bar{\Xi} }{ \alpha_\textrm{c} - 2}  (\pi \lambda)^{\frac{\alpha_\textrm{c}}{2}}  \frac{ \Gamma ( m_\textrm{d} + 1 -\frac{\alpha_\textrm{c}}{2} ) }{\Gamma(m_\textrm{d})}  ,
\label{eq:app:11}
\end{align}
Plugging (\ref{eq:app:11}) into (\ref{eq:app:14}) yields the desired lower bound (\ref{eq:15}).

\subsection{Proof of Proposition \ref{pro:2}}
\label{proof:pro:2}

Using the convexity of the function $\log ( 1 + \frac{1}{x} )$ and applying Jensen's inequality \cite{ngo2011energy},
\begin{align}
R^{(\textrm{c})}_{k} \geq  R^{(\textrm{c,lb})}_{k}  &= \log \left( 1 + \left( \mathbb{E} \left[ \frac{1}{\sinr_k^{(\textrm{c})}}  \right] \right)^{-1}  \right ) \notag \\
&= \log \left( 1 + \left( \mathbb{E} \left[\frac{1}{S^{(\textrm{c})}_k} \right] \cdot ( \mathbb{E} [I^{(\textrm{c}\to \textrm{c})}_k ] + \mathbb{E} [I^{(\textrm{d}\to \textrm{c})}_k ] + N_0 ) \right)^{-1} \right) .
\label{eq:app:08}
\end{align}
In the following three steps, we calculate $\mathbb{E} \left[\frac{1}{S^{(\textrm{c})}_k} \right]$, $\mathbb{E} [I^{(\textrm{c}\to \textrm{c})}_k ]$, and $\mathbb{E} [I^{(\textrm{d}\to \textrm{c})}_k ]$, respectively. Without loss of generality, we assume that $\bw_k^{(\textrm{c})}$ is normalized, i.e., $\|\bw_k^{(\textrm{c})}\|=1$.

\textit{Step 1: calculating $\mathbb{E} \left[\frac{1}{S^{(\textrm{c})}_k} \right]$.} By definition $\| \bw_k^{(\textrm{c})*} \bh_{0k}^{(\textrm{c})} \|^2$ is the squared norm of the projection of the vector $\bh_{0k}^{(\textrm{c})}$ onto the subspace orthogonal to the one spanned by the channel vectors of canceled interferers. The space is of $M - m_\textrm{c} - m_\textrm{d}$ dimensions and is independent of $\bh_{0k}^{(\textrm{c})}$. It follows that $\| \bw_k^{(\textrm{c})*} \bh_{0k}^{(\textrm{c})} \|^2 \sim \chi^2_{2(M - m_\textrm{c} - m_\textrm{d})}$, i.e., $\| \bw_k^{(\textrm{c})*} \bh_{0k}^{(\textrm{c})} \|^2 \sim \Gamma (M - m_\textrm{c} - m_\textrm{d}, 1)$. Therefore, conditioned on $\Xi_{0k}^{(\textrm{c})}$ and $x_{0k}^{(\textrm{c})}$, $\frac{1}{S^{(\textrm{c})}_k} $ is inverse-Gamma distributed and its mean equals
\begin{align}
\mathbb{E} \left[\frac{1}{S^{(\textrm{c})}_k} \right]  = \frac{1}{P_\textrm{c} \Xi_{0k}^{(\textrm{c})} \|x_{0k}^{(\textrm{c})}\|^{-\alpha_\textrm{c}}(M - m_\textrm{c} - m_\textrm{d}-1)}.
\label{eq:app:09}
\end{align}

\textit{Step 2: calculating $\mathbb{E} [I^{(\textrm{c}\to \textrm{c})}_k ]$.} Since $\|\bw_k^{(\textrm{c})}\|=1$ and $\bw_k^{(\textrm{c})}$ is independent of $\bh_{b\ell}^{(\textrm{c})}, \forall \ell \in \mathcal{K}^{(\textrm{c})}_{bk}, \forall b $,  $\bw_k^{(\textrm{c})*} \bh_{b\ell}^{(\textrm{c})}$ is a linear combination of complex Gaussian random variables and thus is distributed as $\mathcal{CN}(0,1)$. It follows that $|\bw_k^{(\textrm{c})*} \bh_{b\ell}^{(\textrm{c})} |^2 \sim \textrm{Exp}(1)$ and
\begin{align}
\mathbb{E} [ I^{(\textrm{c}\to \textrm{c})}_k ] &= \mathbb{E} [\sum_{b=0}^B \sum_{\ell \in \mathcal{K}^{(\textrm{c})}_{bk}  } P_\textrm{c} \Xi_{b\ell}^{(\textrm{c})} \|x_{b\ell}^{(\textrm{c})} \|^{-\alpha_\textrm{c}} |\bw_k^{(\textrm{c})*} \bh_{b\ell}^{(\textrm{c})} |^2 ] =   \sum_{b=0}^B \sum_{\ell \in \mathcal{K}^{(\textrm{c})}_{bk}  } P_\textrm{c} \bar{\Xi} \|x_{b\ell}^{(\textrm{c})} \|^{-\alpha_\textrm{c}} .
\label{eq:app:10}
\end{align}

\textit{Step 3: calculating $\mathbb{E} [I^{(\textrm{d}\to \textrm{c})}_k ]$.} With a similar argument as in Step 2, we have $|\bw_k^{(\textrm{c})*} \bh_i^{(\textrm{d})} |^2 \sim \textrm{Exp}(1)$ and
\begin{align}
\mathbb{E} [ I^{(\textrm{d}\to \textrm{c})}_k ] &= \mathbb{E}_{\Phi} [\sum_{i \in \Phi_k^{(\textrm{c})}}  P_\textrm{d} \Xi_i^{(\textrm{d})} \|x_i^{(\textrm{d})}\|^{-\alpha_\textrm{c}} \mathbb{E}_{\bh} [|\bw_k^{(\textrm{c})*}  \bh_i^{(\textrm{d})} |^2] ] = P_\textrm{d} \bar{\Xi} \mathbb{E}_{\Phi} [\sum_{i \in \Phi_k^{(\textrm{c})}}   \|x_i^{(\textrm{d})}\|^{-\alpha_\textrm{c}}   ] .
\end{align}
The remaining steps for calculating $\mathbb{E} [I^{(\textrm{d}\to \textrm{c})}_k ]$ follow the same steps in the proof of Prop. \ref{pro:4}, i.e., the steps after (\ref{eq:app:13}), and $\mathbb{E} [I^{(\textrm{d}\to \textrm{c})}_k ]$ is given in  (\ref{eq:app:11}).

Finally, plugging  (\ref{eq:app:09}), (\ref{eq:app:10}) and (\ref{eq:app:11}) into (\ref{eq:app:08}) completes the proof.

\subsection{Proof of Lemma \ref{lem:1}}
\label{proof:lem:1}
Since the $M \times T_\textrm{c}$ dimensional noise matrix $\bV_0^{(\textrm{c})}$ consists of i.i.d. $\mathcal{CN}(0,N_0)$ elements and $\bq_{\tilde{k}}^{(\textrm{c})}$ is an orthonomal vector, $\bV_0^{(\textrm{c})} \bq_{\tilde{k}}^{(\textrm{c})}$ also consists of i.i.d. $\mathcal{CN}(0,N_0)$ elements. Similarly, $\bu_{br}^{(\textrm{d})*} \bq_{\tilde{k}}^{(\textrm{c})} \sim \mathcal{CN}(0,1)$.  Using the independence of $\bh_{br}^{(\textrm{d})}, \bu_{br}^{(\textrm{d})*} \bq_k^{(\textrm{c})}$ and $\bV_0^{(\textrm{c})} \bq_{\tilde{k}}^{(\textrm{c})}$, we have $\mathbb{E} [ \tilde{\bv}_k^{(s)}  ]=0$, and
\begin{align}
\mathbb{E} [ \tilde{\bv}_k^{(s)} \tilde{\bv}_k^{(s)*} ] &= \frac{ 1 }{ T_\textrm{c} P_\textrm{c} \Xi^{(s)}_{0k} \|x^{(s)}_{0k}\|^{-\alpha_\textrm{c}} }  \bigg( \sum_{b=0}^{B+1} \sum_{r \in \Phi_{bk}^{(\textrm{c})}} P_\textrm{d}\Xi^{(\textrm{d})}_{br}  \|x^{(\textrm{d})}_{br}\|^{-\alpha_\textrm{c}} \times  \notag \\
&\quad\mathbb{E} [ \bh_{br}^{(\textrm{d})} \bu_{br}^{(\textrm{d})*} \bq_{\tilde{k}}^{(\textrm{c})}   \bq_{\tilde{k}}^{(\textrm{c})*} \bu_{br}^{(\textrm{d})} \bh_{br}^{(\textrm{d})*}  ]   + \mathbb{E} [ \bV^{(\textrm{c})} \bq_{\tilde{k}}^{(\textrm{c})} \bq_{\tilde{k}}^{(\textrm{c})*}  \bV^{(\textrm{c})*} ] \bigg) \notag \\
&= \frac{ 1 }{ T_\textrm{c} P_\textrm{c} \Xi^{(s)}_{0k} \|x^{(s)}_{0k}\|^{-\alpha_\textrm{c}} } \left( \sum_{b=0}^{B+1} \sum_{r \in \Phi_{bk}^{(\textrm{c})}} P_\textrm{d}\Xi^{(\textrm{d})}_{br}  \|x^{(\textrm{d})}_{br}\|^{-\alpha_\textrm{c}}    \mathbb{E} [ \bh_{br}^{(\textrm{d})} \bh_{br}^{(\textrm{d})*}  ]  + N_0 \bI_{M} \right) \notag \\
&= \frac{\sum_{b=0}^{B+1} \sum_{r \in \Phi_{bk}^{(\textrm{c})}} P_\textrm{d}\Xi^{(\textrm{d})}_{br}  \|x^{(\textrm{d})}_{br}\|^{-\alpha_\textrm{c}}    + N_0 }{ T_\textrm{c} P_\textrm{c} \Xi^{(s)}_{0k} \|x^{(s)}_{0k}\|^{-\alpha_\textrm{c}} }  \bI_M.
\end{align}

Further, using that $\bh_{bk}^{(s)}$ consists of i.i.d. $\mathcal{CN}(0,1)$ elements and the independence of $\bh_{bk}^{(s)}$  and $ \tilde{\bv}_k^{(s)}$,
\begin{align}
\mathbb{E} [ \bh_{0k}^{(s)}  \tilde{\by}_k^{(s)*} ] = \mathbb{E} [ \bh_{0k}^{(s)}  \bh^{(s)*}_{0k}  +\sum_{b=1}^B \sqrt{\beta^{(s)}_{bk}} \bh_{0k}^{(s)} \bh^{(s)*}_{bk} + \bh_{0k}^{(s)} \tilde{\bv}_k^{(s)*} ] = \bI_M .
\end{align}
Similarly, we have 
\begin{align}
\mathbb{E} [ \tilde{\by}_k^{(s)}  \tilde{\by}_k^{(s)*} ] &= \mathbb{E} [ \bh_{0k}^{(s)}  \bh^{(s)*}_{0k} + \sum_{b=1}^B  \beta^{(s)}_{bk} \bh_{bk}^{(s)} \bh^{(s)*}_{bk}  + \tilde{\bv}_k^{(s)} \tilde{\bv}_k^{(s)*} ] 
\notag \\
&= \left(1 + \sum_{b=1}^B  \beta^{(s)}_{bk} + \frac{\sum_{b=0}^{B+1} \sum_{r \in \Phi_{bk}^{(\textrm{c})}} P_\textrm{d}\Xi^{(\textrm{d})}_{br}  \|x^{(\textrm{d})}_{br}\|^{-\alpha_\textrm{c}}    + N_0 }{ T_\textrm{c} P_\textrm{c} \Xi^{(s)}_{0k} \|x^{(s)}_{0k}\|^{-\alpha_\textrm{c}} } \right) \bI_M = \frac{1}{\xi^{(s)}_k} \bI_M .
\end{align}
Therefore, the MMSE estimate of $\bh^{(s)}_{0k}$ is
\begin{align}
\hat{\bh}^{(s)}_{0k} &= \mathbb{E} [ \bh_{0k}^{(s)}  \tilde{\by}_k^{(s)*} ] (\mathbb{E} [ \tilde{\by}_k^{(s)}  \tilde{\by}_k^{(s)*} ])^{-1} \tilde{\by}^{(s)}_k =\xi^{(s)}_k \tilde{\by}^{(s)}_k .
\end{align}
Clearly, $\hat{\bh}^{(s)}_{0k}$ is zero mean and its covariance is
$
\mathbb{E}[ \hat{\bh}^{(s)}_{0k} \hat{\bh}^{(s)*}_{0k} ] = \xi^{(s)}_k  \bI_M .
$
As for the estimation error $\bepsilon_k^{(s)} = \bh^{(s)}_{0k} - \hat{\bh}^{(s)}_{0k}$, it is clearly zero mean and its covariance is 
\begin{align}
\mathbb{E}[ \bepsilon_k^{(s)} \bepsilon_k^{(s)*} ] = \mathbb{E} [ \bh_{0k}^{(s)}  \bh^{(s)*}_{0k} ] - \mathbb{E}[ \hat{\bh}^{(s)}_{0k} \hat{\bh}^{(s)*}_{0k} ] = (1-\xi^{(s)}_k) \bI_M .
\end{align}

\subsection{Proof of Proposition \ref{pro:6}}
\label{proof:pro:6}

With D2D links deactivated in the training phase, we have
\begin{align}
&\lim_{M\to \infty} \frac{1}{\sqrt{M}}\left( \frac{1}{M^{1/4}} \sum_{b=0}^B \sqrt{T_\textrm{c} P_\textrm{c} \Xi^{(\textrm{c})}_{bk}}  \|x^{(\textrm{c})}_{bk}\|^{-\frac{\alpha_\textrm{c}}{2}} \bh^{(\textrm{c})}_{bk}    + \bar{\bv}^{(\textrm{c})}_0  \right)^*\by^{(\textrm{c})}_0 
=  \sum_{b=0}^{B}  \sqrt{T_\textrm{c}}  P_\textrm{c} \Xi^{(\textrm{c})}_{b k} \|x^{(\textrm{c})}_{b k}\|^{-\alpha_{\textrm{c}}} u^{(\textrm{c})}_{b k}  \notag \\
&+ \lim_{M\to \infty}  \frac{1}{M^{3/4}} \sum_{b=0}^B \sqrt{T_\textrm{c} P_\textrm{c} \Xi^{(\textrm{c})}_{bk}}  \|x^{(\textrm{c})}_{bk}\|^{-\frac{\alpha_\textrm{c}}{2}}  \sum_{i \in \Phi} \sqrt{P_\textrm{d}\Xi_i^{(\textrm{d})} }\|x^{(\textrm{d})}_i\|^{-\frac{\alpha_{\textrm{c}}}{2}}  \bh^{(\textrm{c})*}_{bk} \bh^{(\textrm{d})}_{i}  u_i^{(\textrm{d})}   \notag \\
&+  \lim_{M\to \infty}  \frac{1}{M^{3/4}} \sum_{b=0}^B \sqrt{T_\textrm{c} P_\textrm{c} \Xi^{(\textrm{c})}_{bk}}  \|x^{(\textrm{c})}_{bk}\|^{-\frac{\alpha_\textrm{c}}{2}} \bh^{(\textrm{c})*}_{bk} {\bv}^{(\textrm{c})}_0  + \lim_{M\to \infty}  \frac{1}{M^{3/4}} \sum_{b=0}^B \sqrt{P_\textrm{c} \Xi^{(\textrm{c})}_{bk}}  \|x^{(\textrm{c})}_{bk}\|^{-\frac{\alpha_\textrm{c}}{2}} \bar{\bv}^{(\textrm{c})*}_0 \bh^{(\textrm{c})}_{bk} u^{(\textrm{c})}_{b k} \notag \\
&+\lim_{M\to \infty}  \frac{1}{\sqrt{M}} \sum_{i \in \Phi} \sqrt{P_\textrm{d}\Xi_i^{(\textrm{d})} } \|x^{(\textrm{d})}_i\|^{-\frac{\alpha_{\textrm{c}}}{2}}    \bar{\bv}^{(\textrm{c})*}_0 \bh^{(\textrm{d})}_{i}  u_i^{(\textrm{d})}  + \lim_{M\to \infty}  \frac{1}{\sqrt{M}}\bar{\bv}^{(\textrm{c})*}_0 {\bv}^{(\textrm{c})}_0.
\label{eq:53}
\end{align}
For the second term on the right hand side of (\ref{eq:53}), we can show that it converges to $0$ in probability by following the same arguments of the proof of Prop. \ref{pro:1}. For the third and fourth terms on the right hand side of (\ref{eq:53}), it is clear that they converge to $0$ almost surely. The last term on the right hand side of (\ref{eq:53}) converges in distribution to a zero-mean complex Gaussian random variable of variance $N_0^2$. The fifth term is zero mean and has variance $\sum_{i \in \Phi} P_\textrm{d}\Xi_i^{(\textrm{d})}  \|x^{(\textrm{d})}_i\|^{-\alpha_{\textrm{c}}} N_0$ but not Gaussian. Using the worst-case noise argument, we conclude (\ref{eq:20}) is achievable.

\bibliographystyle{IEEEtran}
\bibliography{IEEEabrv,Reference}

\end{document}